\documentclass[10pt]{article}
\usepackage[english]{babel}
\usepackage{authblk}
\usepackage[letterpaper,top=2cm,bottom=2cm,left=2cm,right=2cm,marginparwidth=1.75cm]{geometry}

\usepackage{amsmath, amssymb, amsthm, mathrsfs}
\usepackage{graphicx}
\usepackage[colorlinks=true, allcolors=blue]{hyperref}
\usepackage{hyperref}
\usepackage[capitalise]{cleveref}
\usepackage{crossreftools}
\usepackage{color}
\usepackage{cancel}

\usepackage{soul}

\usepackage{layouts}

\def\R{\mathbb{R}}
\def\E{\mathbb{E}}
\newcommand{\dps}{\displaystyle}
\newcommand{\one}{\mathbf{1}}
\newcommand{\partition}{\mathcal{Z}}
\newtheorem{theo}{Theorem}[section]
\newtheorem{prop}[theo]{Proposition}

\newtheorem{rem}[theo]{Remark}

\title{Sampling metastable systems using collective variables and Jarzynski--Crooks paths}
\author[1]{C. Sch\"{o}nle}
\author[1]{M. Gabrié}
\author[2]{T. Lelièvre}
\author[2]{G. Stoltz}
\affil[1]{CMAP, CNRS, École polytechnique, Institut Polytechnique de Paris, 91120 Palaiseau, France}
\affil[2]{CERMICS, Ecole des Ponts, IP Paris, Marne-la-Vallée, France \& MATHERIALS project-team, Inria Paris, France}
\date{}

\begin{document}
\maketitle

\begin{abstract}
We consider the problem of sampling a high dimensional multimodal target probability measure.
We assume that a good proposal kernel to move only a subset of the degrees of freedoms (also known as collective variables) is known a priori. This proposal kernel can for example be built using normalizing flows~\cite{noeBoltzmannGeneratorsSampling2019,liNeuralNetworkRenormalization2018,gabrieAdaptiveMonteCarlo2022}. We show how to extend the move from the collective variable space to the full space and how to implement an accept-reject step in order to get a reversible chain with respect to a target probability measure.
The accept-reject step does not require to know the marginal of the original measure in the collective variable (namely to know the free energy). The obtained algorithm admits several variants, some of them being very close to methods which have been proposed previously in the literature, in particular in~\cite{athenesComputationChemicalPotential2002,nilmeier2011nonequilibrium,chen_enhanced_2015,chen2015generalized,neal2005taking}. We show how the obtained acceptance ratio can be expressed in terms of the work which appears in the Jarzynski--Crooks equality, at least for some variants.  Numerical illustrations demonstrate the efficiency of the approach on various simple test cases and allow us to compare the variants of the algorithm.
\end{abstract}

\section{Introduction}
Metastability poses a major hurdle in establishing the thermodynamical properties of physical systems. In the presence of metastability, the Boltzmann distribution is multimodal and therefore challenging to sample. Indeed, common Markov chain Monte Carlo (MCMC) approaches, such as molecular dynamics, rely on local updates in the space of configuration and fail to converge within an acceptable time as they get trapped in local energy minima. This difficulty is exacerbated as the dimension of the system increases and it becomes less and less likely to transition by chance between the modes.

Two complementary classes of strategies have been developed to tackle the challenge of sampling multimodal distributions. A first class of methods leverages the idea of tempering, combining fast-mixing moves at high-temperature with slow-mixing moves at low-temperature (see for instance \cite{swendsenReplicaMonteCarlo1986,nealAnnealedImportanceSampling2001,delmoralSequentialMonteCarlo2006}, as well as~\cite{woodard_conditions_2009,syed_non-reversible_2022} for recent contributions). In the context of molecular dynamics, this method is called replica exchange and is known to require a careful choice of the numerical parameters (exchange attempt frequency, ladder of temperatures) to be efficient, especially when the dimension increases~\cite{rosta2009error,Henin_Lelievre_Shirts_Valsson_Delemotte_2022}.  A second class of method exploits the observation that the multimodality is often only present along a few degrees of freedom. These degrees of freedom are called collective variables (CV) (or order parameters). Since the CV space is lower dimensional, it is then much simpler to explore and it can be used to bridge modes using enhanced sampling algorithms such as umbrella sampling~\cite{torrie-valleau-77} and the related histogram methods~\cite{SC08,KBSKR92}, thermodynamic integration~\cite{kirkwood-35} or free energy adaptive biasing methods~\cite{darve-pohorille-01,henin-chipot-04,laio-parrinello-02,barducci-bussi-parrinello-08}. Compared to a CV-free method such as parallel tempering, these techniques are known to be less sensitive to the dimension of the problem, at least if the CVs indexing the various modes of the target measure are known. The bottleneck in these approaches thus resides in identifying an appropriate set of CVs, which, in addition, cannot be very high dimensional because these techniques are typically plagued by the curse of dimensionality in the CV space.

Recently, a third category of approaches has emerged based on deep generative models called normalizing flows (NFs) \cite{papamakarios2021}. They allow for efficient sampling of independent realizations and have a tractable probability density. Trained to approach the Boltzmann distribution, they can thus be used to propose candidate configurations which are subsequently reweighted or incorporated in an MCMC sampler such as a Metropolis--Hastings sampler \cite{liNeuralNetworkRenormalization2018,noeBoltzmannGeneratorsSampling2019,gabrieAdaptiveMonteCarlo2022}. Since the proposal samples are generated independently across metastable states, this approach directly addresses the challenge of metastability.
The main price to pay here is the additional cost of training the flow to sufficient accuracy, which is amortized in part by the subsequent performance gain during sampling. However, reaching this level of accuracy becomes increasingly difficult as the dimension grows~\cite{deldebbioEfficientModellingTrivializing2021,greniouxSamplingApproximateTransport2023b,schonleOptimizingMarkovChain2023}, currently placing the limit of applications for these approaches between a few hundreds and a few thousands of degrees of freedom. 

A promising idea is therefore to combine dimensionality reduction and efficient samplers, such as the ones in this third category based on NFs, by applying them in CV-space. On the one hand, targeting a CV space (rather than the full configuration space) reduces the dimensionality and eases the learning of the NF. On the other hand, a NF can drive exploration in a much larger CV-space than traditional enhanced samplers for which CVs are capped at a few dimensions. This is advantageous as capturing the metastability of the system with a representation of a few tens or hundreds of dimensions is significantly simpler. However, sampling directly in CV space is typically impossible since the image measure of the Boltzmann distribution in the CV space (related to the free energy) is intractable. A natural question, and the focus of the present work, is then how CV-configurations proposed by a smart sampler can be employed to build an unbiased sampler in the full configuration space. 

Several works in literature already address this question~\cite{athenesComputationChemicalPotential2002,chen2015generalized,chen_enhanced_2015,neal2005taking,vandecasteele_micro-macro_2023,nilmeier2011nonequilibrium}. We explain the general philosophy behind constructing an unbiased algorithm and describe in detail different versions, some of them corresponding directly to those found in literature. The general idea is always to complete a proposal in CV space to a proposal in full space using a specific protocol. One instance of such an algorithm was given by \cite{vandecasteele_micro-macro_2023} and applied practically in a pseudo-marginal setting in \cite{vandecasteele_pseudo-marginal_2024}. We follow a different approach, relying on a non-equilibrium move in configuration space.
The overall move to the endpoint is accepted or rejected in a Metropolis--Hastings sampler with an acceptance ratio that fulfills detailed balance.
Since it is typically intractable to marginalize over all possible paths between two endpoints, detailed balance is instead considered at the level of moves following specific forward and backward paths between two endpoints -- which implies the traditional detailed balance. 
The works~\cite{athenesComputationChemicalPotential2002, karagiannis_annealed_2013,neal2005taking,chen2015generalized,nilmeier2011nonequilibrium} all fall into this paradigm, focusing on different contexts. A specific choice of non-equilibrium protocol was presented by \cite{neal2005taking}, where the system is dragged from one value of the CV to another by following an interpolation of the target measures between the start and end point. In a Bayesian framework of model selection, this idea was further developed by \cite{karagiannis_annealed_2013} in the context of reversible jump Markov Chain algorithms.
The work \cite{chen_enhanced_2015} explored the possibility to steer the system along a path in CV space with deterministic transformations, which could be further extended to use a combination of deterministic and stochastic steps as seen already in~\cite{vaikuntanathan_escorted_2011}.

The aforementioned algorithms can also be understood from a different angle via the Jarzynski equality~\cite{jarzynski-97}, which provides an unbiased estimator of $\exp(-\beta \Delta F)$, where $\Delta F$ is the free energy difference between two values of the CV. Indeed, following a deterministic schedule between the two CV values of interest combined with a specific equilibration protocol for the remaining coordinates, one obtains a work $\mathcal{W}$ which satisfies $\E[\exp(-\beta \mathcal{W})]=\exp(-\beta \Delta F)$ (at inverse temperature $\beta$). From there, it is natural to formulate a prototypical Metropolis--Hastings algorithm
where, in the acceptance ratio, the exponential of the free energy difference is replaced by the exponential of the work. As we show in this article, some of the previously mentioned algorithms do exactly that: their accept/reject criterion is indeed given by the exponential of a work which precisely fulfills the Jarzynski equality. This point of view offers another intuitive way to understand how these algorithms work. Let us however stress that the reversibility proof relies on the more general Jarzynski--Crooks equality~\cite{Crooks99}  between functionals of the forward and backward paths.

The main contributions of this work are as follows: 
we cast existing algorithms into a common framework and show how they can be viewed through the Jarzynski--Crooks lens (\cref{sec:generalframework,sec:contalgo}). We explicitly describe different variants of the algorithm that can be implemented numerically
and we prove the reversibility of the associated Markov Chain with respect to the target measure (\cref{sec:discretetime}). Through numerical experiments (\cref{sec:numerics}), we provide practical guidance on which algorithm to choose and how to optimize its parameters. Last, we exemplify how proposals are constructed for a concrete model from statistical physics.

In a parallel work \cite{samuel}, the symmetric path algorithm described below is tested on a toy molecular system. There, the framework is extended to allow for the adaptive training of a normalizing flow targeting the CV space along the MCMC procedure.

\section{General framework}
\label{sec:generalframework}

We introduce the main notation in Section~\ref{sec:target_measure} and then the prototype Metropolis--Hastings algorithm we will consider in the following in Section~\ref{sec:algo}. Finally, we provide in Section~\ref{sec:prototype_proof} a general framework to prove the reversibility of the type of Metropolis--Hastings algorithms we will consider.

\subsection{Target measure}\label{sec:target_measure}

Our aim is to sample a Boltzmann measure~$\nu$ over~$\R^d$ of the form
\begin{align}
    \nu({\rm d}x) = \partition^{-1} \exp(-\beta V(x)) \, {\rm d}x,
    \label{eq:boltzmannmeasure}
\end{align}
with $V: \R^d \to \R$ the potential energy function and $\partition$ the normalizing partition function. We assume the prior knowledge of a collective variable~$\xi:\R^d \to \R^\ell$, where $0 < \ell < d$.
For simplicity, we focus on a simple geometry where the collective variable consists of the first $\ell$ degrees of freedom, $\xi(x^1,\dots,x^d) = (x^1,\dots,x^\ell)$, but note that the extension to more complicated (nonlinear) collective variables is possible (see \cref{rem:general_CV}).
Let us denote the remaining coordinates by $x^\perp=(x^{\ell+1}, \ldots, x^d) \in \R^{d-\ell}$ so that a complete microstate can be written as $x=(z,x^\perp)$. The image of the measure $\nu$ by $\xi$ is given by
\begin{align}
    (\xi \# \nu)({\rm d}z) = \nu_{\rm CV}({\rm d}z) = \exp(-\beta F(z)) \, {\rm d}z. 
\end{align}
with the associated free energy
\begin{align}
    F(z)=-\beta^{-1} \ln \left( \int_{\R^{d-\ell}} \partition^{-1} \exp(-\beta V(z,x^\perp)) \, {\rm d}x^\perp\right).
    \label{eq:freeEnergySimple}
\end{align}   
Let us also introduce the family of conditional measures on the variables orthogonal to the CVs:
\begin{align}
\nu_{\perp}({\rm d}x^\perp \, | \, z)=\frac{\partition^{-1}\exp(-\beta V(z, x^\perp))\,{\rm d}x^\perp}{\exp(-\beta F(z))}.
\label{eq:condmeasure}
\end{align}

Collective variables are designed to reduce the dimensionality while capturing the complexity of the target distribution. In particular, a good collective variable should allow to discern the multiple basins in the potential energy landscape of a metastable system. 
Ideally, all the conditional measures $\nu_{\perp}({\rm d}x^\perp \, | \, z)$ should be easy to sample whatever the value of $z$, i.e. should not exhibit metastability, see~\cite{Lelievre13}. We come back to a more precise discussion of what is actually required on the collective variable for the algorithm to be efficient in \cref{sec:app_cv_choice_failure}.
Assuming one has access to a proposal Markov density $Q:\R^\ell \times \R^\ell \to \R_+$ to propose moves in the $\xi$-variable, we are addressing here the question of how to use it to build a sampling algorithm over the whole space~$\R^d$.

\subsection{A prototype algorithm}
\label{sec:algo}
\begin{figure}
    \centering
    \includegraphics{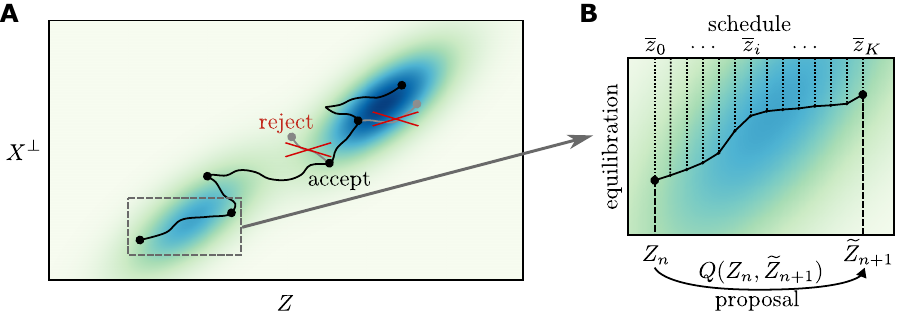}
    \caption{Illustration of the sampling algorithm. Panel A: Starting from a microstate $(Z_n,X^\perp_n)$, a new CV value~$\widetilde Z_{n+1}$ is proposed (Step~1) and a (continuous or discrete) path in the full space to that point is constructed in Step~2. Finally, the move is accepted or rejected in Step~3. Panel B illustrates how to construct such a path with a fixed schedule in $z$, here chosen discrete in time.}
    \label{fig:overview}
\end{figure}
Assume that we are given a proposal Markov density~$Q:\R^\ell \times \R^\ell \to \R_+$ on the collective variable space (for all $z \in \R^\ell$, the measure~$Q(z,\widetilde z) \, {\rm d} \widetilde z$ is a probability distribution on~$\R^\ell$)\footnote{We present the algorithm assuming that the proposal kernel has a density for simplicity; this can be readily extended to more general situations as soon as the Metropolis--Hastings acceptance ratio is well defined.}. 
Since estimating the free energy difference is generally a hard problem, one does not have access to the image measure $\exp(-\beta F(z)) \, {\rm d}z$ and, therefore, one cannot use a simple Metropolis--Hastings algorithm to directly build an MCMC sampler on the collective variables. Instead, the approach we focus on is based on so-called Jarzynski--Crooks moves. They provide an unbiased estimator of the free energy difference between two values of the collective variable via sampled paths in $\R^d$ which follow an imposed schedule in the CV space. The algorithm takes the form of an MCMC sampler, yielding a Markov chain of microstates~$(X_n)_{n \ge 0}$ with~$X_n = (Z_n, X^\perp_n)$, through iterative updates. Let us describe the general idea behind one such MCMC step. Starting from the current state $X_n = (Z_n, X^\perp_n)$ at iteration $n$, the next state is determined in three steps:
\begin{itemize}
    \item (Step 1) Propose a move in the collective variable space:
      \[
      \widetilde{Z}_{n+1} \sim Q(Z_n,\widetilde{z}) \, {\rm d}\widetilde{z},
      \]
    where $Q(z,\widetilde{z}) \in \R^\ell \times \R^\ell \to \R_+$ denotes a Markov transition kernel.
    \item (Step 2) Build a Jarzynski--Crooks path in full space $\R^d$ between collective variable values $Z_n$ and $\widetilde{Z}_{n+1}$. This yields a proposal state $\widetilde{X}_{n+1} = (\widetilde{Z}_{n+1}, \widetilde{X}^\perp_{n+1})$ and a work $\mathcal{W}_{n+1}$ associated with the move. 
\item (Step 3) Draw a random variable $U_{n+1}$ with uniform law over $[0,1]$ and
\begin{itemize}
    \item If 
    \begin{equation}\label{eq:AR}
    U_{n+1} \le \exp(-\beta \mathcal{W}_{n+1})) \displaystyle \frac{Q(\widetilde{Z}_{n+1}, Z_n)}{Q(Z_n,\widetilde{Z}_{n+1})},
    \end{equation}
    then accept the proposal: $X_{n+1}=\widetilde{X}_{n+1}$.
    \item Otherwise, reject the proposal: $X_{n+1} = X_n$.
\end{itemize}
\end{itemize}

The algorithm is schematically visualized in \cref{fig:overview}. We will define Step 2 more precisely later on. For now, let us explain the general intuition behind this update procedure. It is possible to build a path $(\overline X_t)_{0 \le t \le T}$ in full space such that $\overline X_0=X_n$ and $\xi(\overline X_T)=\widetilde{Z}_{n+1}$ (the proposal state being then $\widetilde X_{n+1}=\overline X_T$), associated with a work $\mathcal{W}_{n+1}$ that fulfills the Jarzynski equality: if $X^\perp_n$ is distributed according to $\nu_{\perp}(\cdot|Z_n)$, then, for any function~$\varphi: \R^{d} \to \R$, 
\begin{align}\label{eq:Jarz_consistency}
\E\left(\varphi(\widetilde{X}_{n+1}) \exp(-\beta \mathcal{W}_{n+1}) \, \middle| \, (Z_n,\widetilde Z_{n+1}) \right) = \exp(-\beta (F(\widetilde{Z}_{n+1})-F(Z_n)) \int_{\R^{d-\ell}} \varphi\left(\widetilde Z_{n+1}, x^\perp\right) \nu_{\perp}\left({\rm d}x^\perp \middle| \widetilde Z_{n+1}\right).
\end{align}
On the left hand side, $\widetilde{X}_{n+1}=(\widetilde Z_{n+1}, \widetilde X^\perp_{n+1})$ and the expectation is taken over~$\widetilde X^\perp_{n+1}$, the initial condition~$X^\perp_n$, and the work $\mathcal{W}_{n+1}$, with fixed $Z_n$ and $\widetilde Z_{n+1}$.
In particular, taking $\varphi=1$ in~\eqref{eq:Jarz_consistency}, one recovers the fact that the work provides an estimator for the free energy difference: 
\begin{equation}\label{eq:Jarz_FE}
\E\left[\exp(-\beta \mathcal{W}_{n+1})\, \middle| \, (Z_n,\widetilde Z_{n+1}) \right] = \exp\left[-\beta \left(F(\widetilde{Z}_{n+1})-F(Z_n)\right)\right].
\end{equation}
This motivates the formula used in Step~3 for the Metropolis--Hastings ratio: if the free energy was known, one would just follow a standard Metropolis--Hastings algorithm with acceptance ratio $\exp(-\beta (F(\widetilde{Z}_{n+1})-F(Z_n))) \frac{Q(\widetilde{Z}_{n+1}, Z_n)}{Q(Z_n,\widetilde{Z}_{n+1})}$. Instead, one uses the consistent estimator provided by the work, $\exp(-\beta \mathcal{W}_{n+1})) \frac{Q(\widetilde{Z}_{n+1}, Z_n)}{Q(Z_n,\widetilde{Z}_{n+1})}$. 

It will turn out further on that the equality~\eqref{eq:Jarz_consistency} is not sufficient to prove the reversibility of the Markov chain (as one should actually use the Jarzynski--Crooks equality on path space). Conversely, we will also present algorithms with time-discretized paths which yield reversible Markov chains but fit the framework presented here only in the limit of an infinitely small time step. Still, the above picture provides an intuitive way of understanding the algorithms we will introduce.

In \cref{sec:contalgo}, we present a concrete algorithm using Jarzynski--Crooks moves with a continuous-in-time path schedule and prove that the resulting Markov chains are reversible with respect to the target measure. `Continuous-in-time' refers here to the path built in Step~2 of the algorithm; the overall Markov chain is always discrete.

In \cref{sec:discretetime}, we consider practical algorithms with discrete-in-time paths and prove the reversibility of the resulting Markov chains with respect to the target measure. We show that, under certain conditions, some of them can be equally understood as an instance of the algorithm presented in this section, while for others, one recovers the framework of this section only in the limit of an infinitely small time step.

In all cases, we are focusing on the reversibility of the algorithms, and hence the invariance of the target probability measure. Proving ergodicity would require irreducibility~\cite{MeynTweedie93,DMPS18} -- loosely speaking that any point in target space can be reached from any other point in a finite number of steps. This is not unreasonable to expect for the types of proposal samplers we are interested in. Since we are considering a general `good' proposal sampler $Q$ and are mainly interested in its optimal use, however, we leave a proof of irreducibility for future work.
\subsection{A prototype  proof of reversibility}\label{sec:prototype_proof}

For pedagogical purposes, we give in this section a high-level viewpoint on the essential properties one needs in order for the prototype algorithm presented in the previous section (and the algorithms that will be considered below) to generate a Markov Chain $(X_n)_{n \ge 0}$ which is reversible with respect to $\nu$.

In order to state this general reversibility result, we rewrite the prototype algorithm of the previous section in a slightly more general setting, so that it covers all the variants that will be introduced below. Starting from the current state $X_n = (Z_n, X^\perp_n)$ at iteration $n$, the next state is determined in three steps:
\begin{itemize}
\item (Step 1) Propose a move in the collective variable space:
  $\widetilde{Z}_{n+1} \sim Q(Z_n,\widetilde{z}) \, {\rm d}\widetilde{z}.$
\item (Step 2) Draw the proposal state $\widetilde X_{n+1}^\perp$ and additional random variables $Y_{n+1}$ according to some density~$P_{X_n \to \widetilde Z_{n+1}}$:
  \[
  \left(\widetilde X^\perp_{n+1},Y_{n+1}\right) \sim P_{X_n \to \widetilde Z_{n+1}}(\widetilde x^\perp,y) \, {\rm d} \widetilde x^\perp \, {\rm d}y.
  \]
  Compute the Metropolis--Hastings ratio $R(X_n,\widetilde X_{n+1},Y_{n+1})$, where, for all $x=(z,x^\perp)$, $\widetilde x=(\widetilde z,\widetilde x^\perp)$, and $y$,
  \begin{equation}\label{eq:R}
    R(x,\widetilde x,y)=\frac{\exp(-\beta V(\widetilde x)) P_{\widetilde x\to z}(x^\perp,y)}{\exp(-\beta V(x)) P_{x\to \widetilde z}(\widetilde x^\perp,y)}.
  \end{equation}
\item (Step 3) Draw a random variable $U_{n+1}$ with uniform law over $[0,1]$ and
  \begin{itemize}
  \item If 
    \begin{equation}\label{eq:AR2}
      U_{n+1} \le R\left(X_n,\widetilde X_{n+1},Y_{n+1}\right) \frac{Q(\widetilde{Z}_{n+1}, Z_n)}{Q(Z_n,\widetilde{Z}_{n+1})},
    \end{equation}
    accept the proposal: $X_{n+1}=\widetilde{X}_{n+1}$.
  \item Otherwise, reject the proposal: $X_{n+1} = X_n$.
  \end{itemize}
\end{itemize}

Let us comment on Step 2. Notice first that the introduced density is indexed by $X_n$ and $\widetilde Z_{n+1}$, and that the two sampled random variables $\widetilde X^\perp_{n+1}$ and $Y_{n+1}$ are coupled. The random variable $\widetilde X^\perp_{n+1}$ is used to complete the proposal to $\widetilde X_{n+1}=(\widetilde Z_{n+1},\widetilde X^\perp_{n+1})$. The additional random variable $Y_{n+1}$ contains all the additional noise needed in order to compute the Metropolis--Hastings ratio $R(X_n,\widetilde X_{n+1},Y_{n+1})$ (to make a link with Section~\ref{sec:algo}, compare~\eqref{eq:AR2} with~\eqref{eq:AR}). Notice that a simple consequence of~\eqref{eq:R} is that, for all $z$ and $\widetilde z$,
\[
\exp(-\beta F(z)) \int_{y, x^\perp, \widetilde x^\perp} R(x,\widetilde x,y) P_{x\to \widetilde z}(\widetilde x^\perp,y) \nu_\perp({\rm d} x^\perp | z) \, {\rm d} y \, {\rm d} \widetilde x^\perp = \exp(-\beta F(\widetilde z)),
\]
which shows some consistency of the Metropolis--Hastings ratio $R$ with the standard Metropolis--Hastings ratio that one would use on the collective variable space. This also makes a link with the Jarzynski equality~\eqref{eq:Jarz_FE} mentioned above.

The main result of this section is the following.\footnote{We drew inspiration for this result from discussions with Guanyang Wang (Rutgers University).}
\begin{prop}
  The Markov chain $(X_n)_{n \ge 0}$ generated by the above algorithm is reversible with respect to $\nu$.
\end{prop}

\begin{proof}
  Consider $X_n$ distributed according to $\nu$ and $X_{n+1}$ generated by the algorithm above. The objective is to show that $(X_n,X_{n+1})$ has the same law as $(X_{n+1},X_n)$. Introduce a bounded measurable test function $\varphi:\R^d \times \R^d \to \R$. One has
  \begin{align*}
    \E(\varphi(X_n,X_{n+1})) &=
    \E\left(\varphi(X_n,\widetilde{X}_{n+1}) \one_{U_{n+1}\le \frac{Q(\widetilde{Z}_{n+1}, Z_n)}{Q(Z_n,\widetilde{Z}_{n+1})} R(X_n,\widetilde X_{n+1},Y_{n+1})}\right) \\
    &\quad + \E\left(\varphi(X_{n},X_n)\one_{U_{n+1}> \frac{Q(\widetilde{Z}_{n+1}, Z_n)}{Q(Z_n,\widetilde{Z}_{n+1})}R(X_n,\widetilde X_{n+1},Y_{n+1})}\right) \\
    &= \E\left(\varphi(X_n,\widetilde{X}_{n+1}) \left[1 \wedge \left( \frac{Q(\widetilde{Z}_{n+1}, Z_n)}{Q(Z_n,\widetilde{Z}_{n+1})}R(X_n,\widetilde X_{n+1},Y_{n+1})\right)\right]\right)\\
    &\quad + \E\left(\varphi(X_n, X_{n}) \left[ 1-  1 \wedge \left(\frac{Q(\widetilde{Z}_{n+1}, Z_n)}{Q(Z_n,\widetilde{Z}_{n+1})}R(X_n,\widetilde X_{n+1},Y_{n+1})\right)\right]\right).
  \end{align*}
  The reversibility of the Markov chain follows from the fact that the sum of the latter two terms is equal to~$\E(\varphi(X_{n+1},X_n))$ for all test functions~$\varphi$, which is the case if and only if
  \begin{align}
    & \E\left(\varphi(X_n,\widetilde{X}_{n+1}) \left[1 \wedge \left( \frac{Q(\widetilde{Z}_{n+1}, Z_n)}{Q(Z_n,\widetilde{Z}_{n+1})}R(X_n,\widetilde X_{n+1},Y_{n+1})\right)\right]\right)\notag\\
    & \qquad = \E\left(\varphi(\widetilde{X}_{n+1},X_n) \left[1 \wedge \left( \frac{Q(\widetilde{Z}_{n+1}, Z_n)}{Q(Z_n,\widetilde{Z}_{n+1})}R(X_n,\widetilde X_{n+1},Y_{n+1})\right)\right]\right).
    \label{eq:symmetry}
  \end{align}
  Notice that (remember the notation $x=(z,x^\perp)$ and $\widetilde x=(\widetilde z,\widetilde x^\perp)$)
  \begin{align*}
    &\partition \, 
    \E\left(\varphi(X_n,\widetilde{X}_{n+1}) \left(1 \wedge \left( \frac{Q(\widetilde{Z}_{n+1}, Z_n)}{Q(Z_n,\widetilde{Z}_{n+1})}R(X_n,\widetilde X_{n+1},Y_{n+1})\right)\right)\right)\\
    &=\int_{x,\widetilde x,y} \varphi(x,\widetilde x)
    \left(1 \wedge \left( \frac{Q(\widetilde z, z)}{Q(z,\widetilde z)}R(x,\widetilde x,y)\right)\right) \exp(-\beta V(x)) Q(z,\widetilde z) P_{x \to \widetilde z}(\widetilde x^\perp,y) \, {\rm d}x \, {\rm d}\widetilde x \, {\rm d}y\\
    &=\int_{x,\widetilde x,y} \varphi(x,\widetilde x)
    \left[ \exp(-\beta V(x)) Q(z,\widetilde z) P_{x \to \widetilde z}(\widetilde x^\perp,y)\right] \wedge  \left[ \exp(-\beta V(x))  P_{x \to \widetilde z}(\widetilde x^\perp,y)Q(\widetilde z, z)R(x,\widetilde x,y)\right] {\rm d}x \, {\rm d}\widetilde x \, {\rm d}y\\
    &=\int_{x,\widetilde x,y} \varphi(x,\widetilde x)
    \left[ \exp(-\beta V(x)) Q(z,\widetilde z) P_{x \to \widetilde z}(\widetilde x^\perp,y)\right] \wedge  \left[\exp(-\beta V(\widetilde x))  P_{\widetilde x \to  z}(x^\perp,y)Q(\widetilde z, z)\right] {\rm d}x \, {\rm d}\widetilde x \, {\rm d}y,
  \end{align*}
  where we used~\eqref{eq:R} to obtain the last equality. It is easy to check that the function
  \[
  (x,\widetilde x) \mapsto \int_{y} \left[ \exp(-\beta V(x)) Q(z,\widetilde z) P_{x \to \widetilde z}(\widetilde x^\perp,y) \right] \wedge \left[ \exp(-\beta V(\widetilde x)) Q(\widetilde z,z ) P_{\widetilde x \to  z}(x^\perp,y)\right] {\rm d}y
  \]
  is invariant under the exchange of~$x=(z,x^\perp)$ and~$\widetilde x=(\widetilde{z},\widetilde{x}^\perp)$. This shows that~\eqref{eq:symmetry} holds, and thus concludes the proof.
\end{proof}

In some sense, the aim of this work is to propose and numerically test the efficiency of various algorithms which enters the framework of the generic algorithm described in this section, by building couples of functions~$(P_{x \to \widetilde z},R(x,\widetilde x,y))$ such that~\eqref{eq:R} holds.

\begin{rem} The algorithm studied here is related to the so-called ``Stateless auxiliary variable MCMC''.\footnote{We would like to thank Andi Wang (University of Warwick) for pointing this out to us.} It enters the framework of~\cite{andrieu2020general}; see in particular~\cite[Proposition~1]{andrieu2020general} with~$\phi(x,y,\widetilde x)=(\widetilde x,y,x)$. It is also related to the so-called ``super-detailed balance'' of~\cite{frenkelSpeedupMonteCarlo2004}, which also inspired the proof of reversibility in \cite{nilmeier2011nonequilibrium}. 
\end{rem}

\section{The continuous-in-time path algorithm}
\label{sec:contalgo}

We now present a way of constructing a proposal $\widetilde{X}_{n+1}$ with associated work $\mathcal{W}_{n+1}$ for Step~2 of the prototype algorithm from \cref{sec:algo}. Here, we build a path with continuous-in-time dynamics and prove that the resulting overall algorithm is reversible with respect to the target measure.

\subsection{Construction of the constrained path}\label{sec:cont_path}

Let us denote the initial and final points in CV space by $z=Z_n$ and $\widetilde z=\widetilde{Z}_{n+1}$. From these two points, we build a deterministic path $(z(t))_{0 \le t \le T}$ in the CV space which is fully determined by the end points $(z,\widetilde z)$ and which satisfies the two following properties: $t \mapsto z(t)$ is $\mathcal C^1$ and $(z(0),z(T))=(z,\widetilde z)$. One could for example consider a linear path along a straight line, $z(t)=(1-t/T)z+ (t/T) \widetilde z$, where the duration $T$ is a deterministic function of $(z,\widetilde z)$.

The full path $(\overline{X}_t)_{0 \le t \le T}$ with $\overline{X}_t=(\overline Z_t, \overline X^\perp_t)$ is then built, starting from $\overline X_0 = X_n$, following the constrained overdamped Langevin dynamics. In the `true' reaction coordinate case where the collective variable $\xi$ is not just a subset of all coordinates, special care has to be taken to ensure that $\overline{Z}_t=\xi(\overline{X}_t)$ at all times. In our simple case the dynamics simply read 
\begin{equation}\label{eq:forward_process}
\left\{
\begin{aligned}
   {\rm d}\overline Z_t&=z'(t) \,{\rm d}t,\\
   {\rm d} \overline X^\perp_t&=-\nabla^\perp V(\overline X_t) \, {\rm d}t + \sqrt{2 \beta^{-1}}\, {\rm d} W_t.
\end{aligned}
\right. 
\end{equation}
The gradient is defined as $\nabla^\perp V=(\partial_{x_{\ell+1}} V, \ldots, \partial_{x_d} V)^\top$ and $(W_t)$ is a $(d-\ell)$-dimensional Brownian motion. Since the initial condition is $\overline Z_0 = z(0) = z$, one has $\overline Z_t=z(t)$ for all $t \in [0,T]$, and in particular $\overline Z_T=z(T)=\widetilde z$. The associated work writes:
\begin{align}
\mathcal{W}=\int_0^T \left\langle \nabla_z V(\overline X_t), z'(t)\right\rangle \, {\rm d}t,
\label{eq:jarzwork}
\end{align}
with $\nabla_z V = (\partial_{x_1}V,\, \dots, \partial_{x_\ell}V)^T$.
The endpoint of the constructed path is then set as the proposed move in Step~2 of the algorithm $\widetilde{X}_{n+1}=\overline X_T$ and the work $\mathcal{W}$ is used as $\mathcal{W}_{n+1}$ in Step~3 of the algorithm presented in \cref{sec:algo}; see~\eqref{eq:AR}.

We prove in \cref{sec:reversibility_continuous} that the resulting overall sampling algorithm is reversible with respect to the target measure. The proof uses the so-called Jarzynski--Crooks equality~\cite{jarzynski-97,Crooks99}, presented in \cref{sec:JarzCrook}. 

\subsection{Jarzynski and Jarzynski--Crooks equalities}
\label{sec:JarzCrook}

We assume in this section that $\overline X^\perp_0 = X^\perp_n$ is drawn from the equilibrium conditional measure
\begin{equation}\label{eq:Jarz_IC_eq}
\overline X_0^\perp \sim \nu_{\perp}({\rm d}x^\perp | z),
\end{equation}
see \eqref{eq:condmeasure}. Let us now give two standard results concerning the non-equilibrium path $(\overline X_t)_{0 \le t \le T}$.
The first result relates the works to the free energy difference~\cite{jarzynski-97}. 

\begin{prop}[Jarzynski equality]\label{prop:Jarz}
Consider the constrained process $(\overline X_t)_{0 \le t \le T}$ with initial condition satisfying~\eqref{eq:Jarz_IC_eq}.
Then, for any bounded measurable function $\varphi: \R^d \to \R$,
\[
\E\left[\varphi(\overline X_T) \exp(-\beta \mathcal{W})\right] = \exp(-\beta (F(\widetilde z)-F(z))) \int_{\R^{d-\ell}} \varphi(\widetilde z, x^\perp) \nu_{\perp}({\rm d} x^\perp \,|\,\widetilde z).
\]
\end{prop}

To prove the consistency of the algorithm, we need a stronger result than Proposition~\ref{prop:Jarz}: the Jarzynski--Crooks generalized work fluctuation identity. This requires introducing the so-called backward process associated with the forward process~\eqref{eq:forward_process}. Let $(\overline X^{\rm b}_t)_{0 \le t \le T}$ satisfy (compare with~\eqref{eq:forward_process})
\begin{equation}\label{eq:backward_process}
\left\{
\begin{aligned}
    {\rm d}\overline Z^{\rm b}_t &=-z'(T-t) \,{\rm d}t\\
    {\rm d}\overline X^{\rm b, \perp}_t&=-\nabla^\perp V(\overline Z^{\rm b}_t,\overline X^{\rm b, \perp}_t) \, {\rm d}t + \sqrt{2 \beta^{-1}} \, {\rm d}W_t,
\end{aligned}
\right.
\end{equation}
with initial condition $\overline X^{\rm b}_0=(\widetilde z,\overline X^{\rm b, \perp}_0)$ where
\begin{equation}\label{eq:eq:Jarz_ICb_eq}
  \overline X^{\rm b, \perp}_0 \sim \nu_{\perp}(d \overline x^\perp|\widetilde{z}). 
\end{equation}
Notice that since $\overline Z^{\rm b}_0=z(T)=\widetilde z$, the path in CV space will just be the time-reversed version of the forward path, namely~$\overline{Z}_t^{\rm b} = z(T-t)$ (in particular, $\overline Z^{\rm b}_T=z(0)$).
The full path is denoted by $\overline X^b=(\overline X^{\rm b}_t)_{0 \le t \le T}$, where for all $t \in [0,T]$, $\overline X^{\rm b}_t=(\overline Z^{\rm b}_t,\overline X^{\rm b, \perp}_t) \in \R^\ell \times \R^{d-\ell}$. Then one has the following result~\cite{Crooks99}.

\begin{prop}[Jarzynski--Crooks equality]
  \label{prop:Jarz-Crooks}
Consider the forward process $\overline X = (\overline X_t)_{0 \le t \le T}$ satisfying~\eqref{eq:forward_process}--\eqref{eq:Jarz_IC_eq} and the associated backward process $\overline X^{\rm b} = (\overline X^{\rm b}_t)_{0 \le t \le T}$ satisfying~\eqref{eq:backward_process}--\eqref{eq:eq:Jarz_ICb_eq}.
Then, for any (bounded measurable) path functional $\varphi_{[0,T]}: {\mathcal C}([0,T],\R^d) \to \R$,
$$\exp(-\beta (F(\widetilde z)-F(z))) \E\left(\varphi^{\rm r}_{[0,T]}(\overline X^{\rm b})\right) = \E\left(\varphi_{[0,T]}(\overline X) \exp(-\beta \mathcal{W})\right),$$
where $\varphi^{\rm r}_{[0,T]}$ denotes the map $\varphi$ composed with the time reversal operator: for any path $(x_t)_{0 \le t \le T}$,
$$\varphi^{\rm r}_{[0,T]}( (x_t)_{0 \le t \le T})=\varphi_{[0,T]}((x_{T-t})_{0 \le t \le T}).$$
\end{prop}

We refer to~\cite[Theorems~4.10 and~4.19]{lelievre-rousset-stoltz-book-10} for a proof of this result.
Let us emphasize that Proposition~\ref{prop:Jarz} is a corollary of Proposition~\ref{prop:Jarz-Crooks}, by choosing $\varphi_{[0,T]}((x_t)_{0 \le t \le T})=\varphi(x_T)$.

\subsection{The reversibility result}
\label{sec:reversibility_continuous}

Using \cref{prop:Jarz-Crooks}, we prove the following result (see \cref{proof:th:reversibility} for the proof). 
\begin{theo}\label{th:reversibility}
  Consider the algorithm from Section~\ref{sec:algo} with the continuous-in-time Jarzynski--Crooks move made precise in Section~\ref{sec:cont_path}. Assume that the schedule associated with the move from~$\widetilde z$ to~$z$ is the time reversed schedule associated with the move from~$z$ to~$\widetilde z$, \emph{i.e.}, for all $(z,\widetilde z) \in \R^2$,
  \begin{equation}\label{eq:hyp_rev_z}
   z^{(\widetilde z, z)}(t)= z^{(z, \widetilde z)}(T-t),
 \end{equation}
where we explicitly indicate the dependence of the schedule $(z(t))_{0 \le t \le T}$ on the endpoints in superscript. Then, the Markov chain $X_n = (Z_n, X^\perp_n)_{n \ge 0}$ is reversible with respect to the probability measure~$\nu$. In particular, it admits~$\nu$ as an invariant probability measure.
\end{theo}

\begin{rem}\label{rem:general_CV}
    As mentioned above, in all this work, we focus on the case where the collective variable is simply the projection onto the first $\ell$ coordinates. In practice, nonlinear collective variables often come into play. All the algorithms we present can be generalized to the setting of a general collective variable, the only difficulty being to replace the path generating procedure~\eqref{eq:forward_process} by dynamics in the full space which are constrained to follow the prescribed schedule on the collective variable. We refer for example to~\cite{lelievre-rousset-stoltz-07-a} and~\cite[Section~3.2]{lelievre-rousset-stoltz-book-10} for details on how such algorithms can be built, for both continuous-time paths and time-discretized paths.
\end{rem}

\section{Discrete-in-time path algorithms}
\label{sec:discretetime}
In practice, we need to discretize the continuous dynamics from Section~\ref{sec:contalgo} to construct a path. The general idea is to follow an imposed schedule in CV space from $Z_n$ to $\widetilde{Z}_{n+1}$ which is intertwined with equilibration steps of the remaining coordinates $x^\perp$. As was noted by \cite{nilmeier2011nonequilibrium}, one can define unbiased sampling schemes at the discrete level by fulfilling a detailed balance condition of type \eqref{eq:R} involving the probability to propose $(\widetilde{Z}_{n+1}, \widetilde X_{n+1}^\perp)$ starting from $(Z_n, X^\perp_n)$ and the probability to propose the reverse move, following related forward and backward paths. 
We stress that even though one might say that this concerns the probability of proposing a certain path, this algorithm has nothing to do with sampling in path space (path sampling) since we are not dealing with path measures and are just ultimately interested in the distributions of the endpoints of the constructed paths. The general approach is somewhat similar to Hamiltonian Monte Carlo (HMC)~\cite{DUANE1987216} in that it follows specific dynamics to generate a proposal update (with the particularity for HMC that, due to the exact reversibility of the dynamics, the intermediate trajectory values do not enter into the acceptance criterion).  

A variety of algorithms can be derived as discrete-in-time analogs to the algorithm discussed in \cref{sec:contalgo}. Here, we present three different algorithms schematically represented in \cref{fig:SketchAlgos}, respectively in Sections~\ref{sec:Nilmeier},~\ref{sec:ChenRoux} and~\ref{sec:symm}. They differ in the exact succession of movements in collective variable space and equilibration steps and, in this way, which forward and backward paths are considered for the detailed balance condition. For each algorithm, we show reversibility of the corresponding Markov chain with respect to the target Boltzmann measure.

The first algorithm, inspired by \cite{nilmeier2011nonequilibrium}, uses a protocol where the backward path is different from the forward path, which is why we call it the `asymmetric algorithm'. The second is a generalization of \cite{chen2015generalized}. It uses a protocol with a random coin flip to decide between two paths, which makes its backward and forward paths symmetric, and hence we term it `stochastically symmetric'. We also introduce a third algorithm with an additional equilibration step, where the backward and forward paths are the same, which motivates the name `symmetric algorithm'. Such a symmetric strategy was already considered in \cite{athenesComputationChemicalPotential2002}. As it will turn out, given that the kernels in the equilibration step are appropriately chosen, the symmetrization of the forward and backward paths allows to exactly re-frame the last two algorithms in the context of the prototypical algorithm of \cref{sec:algo}, including a well-defined notion of a `work' for Step~3. This can be understood as a natural way to avoid that the acceptance degrades with the number of steps of the path.

All algorithms share a very similar structure. To avoid any ambiguity, we present them and prove their reversibility separately, even though the proofs follow a very similar pattern, as already highlighted in Section~\ref{sec:prototype_proof}. For the stochastically symmetric and the symmetric algorithms, we show how they constitute instances of our prototype algorithm of \cref{sec:algo}.

\begin{figure}
  \centering
  \includegraphics{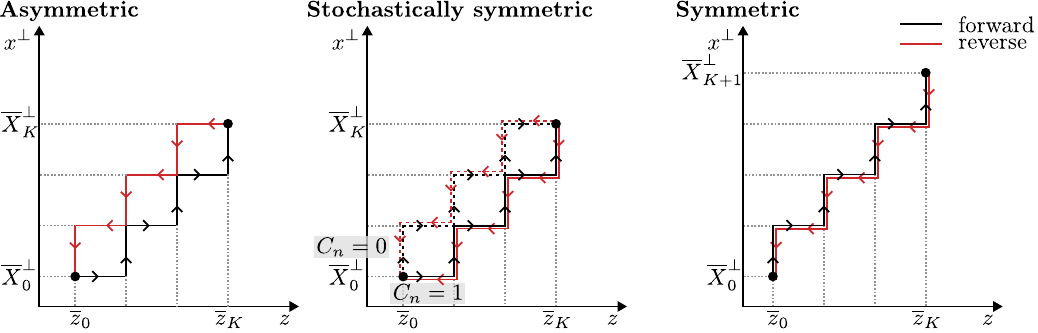}
  \caption{Schematic representations of how paths are built in the presented discrete algorithms between CV coordinate values $Z_n=\overline z_0$ and $\widetilde Z_{n+1}=\overline z_K$. Red lines designate the reverse path being considered in the acceptance ratio. This figure assumes that the schedule satisfies the symmetry property~\eqref{eq:rev_chemin_prime}.}
    \label{fig:SketchAlgos}
\end{figure}

\subsection{The asymmetric path algorithm}\label{sec:Nilmeier}
This first algorithm is inspired by what is proposed in~\cite{nilmeier2011nonequilibrium}. At iteration $n$, let us consider the current state $X_n=(Z_n,X_n^\perp)$. Given a Markov kernel $Q(z,\widetilde{z}) \in \R^\ell \times \R^\ell \to \R_+$ on CV space, and a CV-dependent Markov kernel\footnote{Again, we assume here that the kernels have  densities for simplicity but extensions to more general kernels are straightforward.} on the remaining degrees of freedom $P_{z}(x^\perp,{\widetilde{x}}^\perp)$, the Markov kernel generating the next state $X_{n+1}=(Z_{n+1},X_{n+1}^\perp)$ is determined in three steps:
\begin{itemize}
    \item (Step 1) Propose a move in the collective variable space: $\widetilde{Z}_{n+1} \sim Q(Z_n,\widetilde{z}) \, {\rm d}\widetilde{z}$.
    \item (Step 2) \begin{itemize}
\item Introduce a time-discrete schedule\footnote{For simplicity, $K$ is here and thereafter assumed to be fixed, but a generalization to a number of steps which depends on $(Z_n,\widetilde Z_{n+1})$ is straightforward.} $(\overline z_0=Z_n, \overline z_1, \ldots, \overline z_{K-1}, \overline z_K=\widetilde{Z}_{n+1})$ which is deterministically defined as a function of $(Z_n,\widetilde{Z}_{n+1})$. Here and in the sequel, we use the notation~$\overline z_k(Z_n,\widetilde Z_{n+1})$ for $k \in \{0, \ldots, K\}$ to make explicit the dependence on~$(Z_n,\widetilde Z_{n+1})$.
\item Generate a sequence of states $(\overline X^\perp_{k})_{0 \le k \le K}$ as~$\overline X^\perp_0=X_n^\perp$ and~$\overline X^\perp_{k+1}$ sampled with respect to $P_{\overline z_{k+1}}(\overline X_k^\perp,x^\perp)\mathrm{d}x^\perp$ for $k=0,\ldots,K-1$.
\item Set $\widetilde X^\perp_{n+1}=\overline X^\perp_{K}$. \end{itemize}
\item (Step 3) Draw a random variable $U_{n+1}$ with uniform law over $[0,1]$ and
\begin{itemize}
    \item If $\displaystyle U_{n+1} \le  \frac{\dps \exp(-\beta V(\widetilde Z_{n+1},\widetilde X_{n+1}^\perp)) Q(\widetilde{Z}_{n+1}, Z_n) \prod_{k=1}^K P_{\overline z_k(\widetilde Z_{n+1}, Z_n )}(\overline X^\perp_{K-k+1},\overline X^\perp_{K-k})}{\dps \exp(-\beta V(Z_n,X_n^\perp)) Q(Z_n,\widetilde{Z}_{n+1}) \prod_{k=1}^K P_{\overline z_k(Z_n,\widetilde Z_{n+1})}(\overline X^\perp_{k-1},\overline X^\perp_{k})}$,\\ accept the proposal: $(Z_{n+1},X^\perp_{n+1})=(\widetilde Z_{n+1},\widetilde X^\perp_{n+1})$.
    \item Otherwise, reject the proposal: $(Z_{n+1},X^\perp_{n+1})=(Z_{n}, X^\perp_{n})$.
\end{itemize}
\end{itemize}
The acceptance ratio in Step~3 involves in its denominator the probability of a backward path that follows the same order of successive blocks of one CV update followed by one equilibration update, resulting in a mirroring path (see the left panel in \cref{fig:SketchAlgos}), hence the denomination `asymmetric algorithm'. 

The reversibility of the `asymmetric algorithm' is ensured by the following result, proved in~ \cref{proof:eq:PropNilmeier}. 
\begin{prop}
\label{eq:PropNilmeier}
    The Markov chain $(X_n)_{n \ge 0}$ generated by the above algorithm is reversible with respect to $\nu({\rm d}x) = \partition^{-1} \exp(-\beta V(z,x^\perp)) \, {\rm d}z \, {\rm d}x^\perp$.   
\end{prop}

Let us now make three remarks which will actually also be relevant for the next two algorithms presented in the following two sections. First, the proof does not require that the kernel $P_{z}(x^\perp, \widetilde x^\perp) \, {\rm d} \widetilde x^\perp$ leaves $\nu_\perp({\rm d}x^\perp|z)$  invariant. Kernels with different invariant measures may therefore be used, which constitutes one opportunity of optimization. This is however not the focus of the present work and is left for future study.

Second, the proof does not require either that the proposed schedule is reversed in time upon inverting the end points $Z_n$ and $\widetilde Z_{n+1}$. It could be any deterministic schedule. However, in order for the acceptance probability not to be too small, we would like $\prod_{k=1}^K P_{\overline z_k(\widetilde Z_{n+1}, Z_n )}(\overline X^\perp_{K-k+1},\overline X^\perp_{K-k}) $ not to be negligible compared to $\prod_{k=1}^K P_{\overline z_k(Z_n,\widetilde Z_{n+1})}(\overline X^\perp_{k-1},\overline X^\perp_{k})$.

For this, and this is the third remark, it is enlightening to rewrite the acceptance ratio assuming a symmetric schedule, namely that (this is the time-discrete equivalent of~\cref{eq:hyp_rev_z})
\begin{equation}
  \label{eq:rev_chemin_prime}
  \forall (z,\widetilde z), \quad \forall k \in \{0, \ldots, K\}, \qquad \overline z_k(z,\widetilde z)=\overline z_{K-k}(\widetilde z,z).
\end{equation}
In this case, the last factor of the acceptance ratio in Step~3 can be rewritten as 
\begin{align*}
    \prod_{k=1}^K \frac{P_{\overline z_{k-1}(Z_n,\widetilde Z_{n+1})}(\overline X^\perp_{k},\overline X^\perp_{k-1})}{P_{\overline z_k(Z_n,\widetilde Z_{n+1})}(\overline X^\perp_{k-1},\overline X^\perp_{k})}.
\end{align*}
A high acceptance rate can be attained if the individual ratios in the product for each $k$ are close to 1. Yet, the ratio of probabilities with different laws for the forward and backward paths appears (note the shift in $\overline z$ schedule indices) and this is, again, why we call this the `asymmetric algorithm'. This is represented in the left panel of \cref{fig:SketchAlgos}: the law of the backward path (in red) is not the law of the forward path (in black). 
Under the symmetric schedule assumption~\eqref{eq:rev_chemin_prime}, this algorithm follows the non-equilibrium candidate Monte Carlo framework of \cite{nilmeier2011nonequilibrium} with CV space updates corresponding to `perturbations' and equilibration steps corresponding to `propagations'.

\subsection{The stochastically symmetric (or coin) path algorithm}\label{sec:ChenRoux}
The second algorithm can be seen as a generalization of the algorithm introduced in~\cite{chen2015generalized}. In the word of the authors of the latter, the law over protocols now explicitly involves choosing ``a protocol that starts with propagation or perturbation with equal probability".
Due to this kind of coin flip we also use the shorthand denomination `coin' algorithm in the following. In the algorithm presented below, the forward and backward paths are exactly the same if the schedule in CV space is symmetric \eqref{eq:rev_chemin_prime}. 

At iteration $n$, let us consider the current state $X_n=(Z_n,X_n^\perp)$. Given the same ingredients as in Section~\ref{sec:Nilmeier}, the iteration then proceeds in three steps:
\begin{itemize}
\item (Step 1) Propose a move in the collective variable space: $\widetilde{Z}_{n+1} \sim Q(Z_n,\widetilde{z}) \, {\rm d}\widetilde{z}$.
\item (Step 2)
  \begin{itemize}
  \item Introduce a time-discrete schedule $(\overline z_0=Z_n, \overline z_1, \ldots, \overline z_{K-1}, \overline z_K=\widetilde{Z}_{n+1})$ which is deterministically defined as a function of $(Z_n,\widetilde{Z}_{n+1})$. 
  \item Draw a Bernoulli random variable $C_n \in \{0,1\}$ with parameter~$1/2$.
  \item Generate a sequence of states $(\overline X^\perp_{k})_{0 \le k \le K}$ as $\overline X^\perp_0=X_n^\perp$ and~$\overline X^\perp_{k+1}$ sampled with respect to $P_{\overline z_{k+C_n}}(\overline X_k^\perp,x^\perp) \, {\rm d} x^\perp$ for~$k=0,\ldots,K-1$. 
  \item Set $\widetilde X^\perp_{n+1}=\overline X^\perp_{K}$.
  \end{itemize}
\item (Step 3) Draw a random variable $U_{n+1}$ with uniform law over $[0,1]$ and
  \begin{itemize}
  \item If $\displaystyle U_{n+1} \le  \frac{\dps \exp(-\beta V(\widetilde Z_{n+1},\widetilde X_{n+1}^\perp)) Q(\widetilde{Z}_{n+1}, Z_n) \prod_{k=0}^{K-1} P_{\overline z_{k+1-C_n}(\widetilde Z_{n+1}, Z_n )}(\overline X^\perp_{K-k},\overline X^\perp_{K-k-1})}{\dps \exp(-\beta V(Z_n,X_n^\perp)) Q(Z_n,\widetilde{Z}_{n+1}) \prod_{k=1}^{K} P_{\overline z_{k+C_n-1} (Z_n,\widetilde Z_{n+1})}(\overline X^\perp_{k-1},\overline X^\perp_{k})}$,\\ accept the proposal: $(Z_{n+1},X^\perp_{n+1})=(\widetilde Z_{n+1},\widetilde X^\perp_{n+1})$.
    \item Otherwise, reject the proposal: $(Z_{n+1},X^\perp_{n+1})=(Z_{n}, X^\perp_{n})$.
\end{itemize}
\end{itemize}
The random variable $C_n$ determines whether the path will start by an update of $\overline Z$ or of $\overline X^\perp$ (see \cref{fig:SketchAlgos}). Formally, this choice is reflected by the shift in the index of the schedule $\overline z_k$ to consider in the equilibration kernel of $\overline X^\perp$ involved in each factor of the acceptance ratio.

\begin{prop}
  \label{eq:PropChenRoux}
  The Markov chain $(X_n)_{n \ge 0}$ generated by the above algorithm is reversible with respect to $\nu({\rm d}x) = \partition^{-1} \exp(-\beta V(z,x^\perp)) \, {\rm d}z \, {\rm d}x^\perp$.   
\end{prop}

The proof is given in \cref{proof:eq:PropChenRoux} and again it does not require any specific assumption on the kernels, nor any symmetry property on the schedule in~$z$ when interchanging the end points.

It is again enlightening to consider this algorithm assuming that the schedule satisfies the symmetry property~\eqref{eq:rev_chemin_prime} to ensure a non-vanishing acceptance as~$K$ grows, as was done by \cite{chen2015generalized}.
The name given to this algorithm derives from the fact that, assuming \eqref{eq:rev_chemin_prime}, the final factor of the acceptance ratio in Step~3 can be rewritten as 
\begin{align*}
    \prod_{k=1}^{K} \frac{P_{\overline z_{k+C_n-1} (Z_n,\widetilde Z_{n+1})}(\overline X^\perp_{k},\overline X^\perp_{k-1})}{P_{\overline z_{k+C_n-1} (Z_n,\widetilde Z_{n+1})}(\overline X^\perp_{k-1},\overline X^\perp_{k})},
\end{align*}
which in the numerator involves the likelihood of the forward path followed backwards in time, for the successive forward kernels (this is illustrated in \cref{fig:SketchAlgos}: conditionally on $C_n$, the red path is exactly the black path followed backwards).

In order to obtain an even more interpretable formula of the acceptance ratio, let us assume additionally that the equilibration kernels are reversible with respect to the family of conditional target measures $\nu_\perp({\rm d} x^\perp |z) $, namely that
\begin{equation}\label{eq:rev_second}
  \forall z \in \R^d, \qquad  P_{z}(x^\perp, {\rm d}\widetilde x^\perp)\exp(-\beta V(z,x^\perp)) \, {\rm d}x^\perp  = P_{z}(\widetilde x^\perp, {\rm d} x^\perp)\exp(-\beta V(z,\widetilde x^\perp)) \, {\rm d}\widetilde x^\perp.
\end{equation}
Equilibration kernels satisfying that property can for example be built using a Metropolis algorithm, for instance one step of the Metropolis Adjusted Langevin Algorithm (MALA)~\cite{RDF78,roberts1996} with target measure $\nu_\perp({\rm d} x^\perp |z)$. In that case, one can directly link the acceptance ratio to the Jarzynski--Crooks works and understand it as an instance of the prototypical algorithm of \cref{sec:algo}, as stated in the following proposition.

\begin{prop}
\label{eq:propchenrouxwork}
    Under Assumptions~\eqref{eq:rev_chemin_prime} and~\eqref{eq:rev_second}, it holds
    \begin{equation}\label{eq:work}
      \frac{\dps \exp(-\beta V(\widetilde Z_{n+1},\widetilde X_{n+1}^\perp)) Q(\widetilde{Z}_{n+1}, Z_n) \prod_{k=0}^{K-1} P_{\overline z_{k+1-C_n}(\widetilde Z_{n+1}, Z_n )}(\overline X^\perp_{K-k},\overline X^\perp_{K-k-1})}{\dps \exp(-\beta V(Z_n,X_n^\perp)) Q(Z_n,\widetilde{Z}_{n+1}) \prod_{k=1}^{K} P_{\overline z_{k+C_n-1} (Z_n,\widetilde Z_{n+1})}(\overline X^\perp_{k-1},\overline X^\perp_{k})}= \frac{Q(\widetilde{Z}_{n+1}, Z_n)}{Q( Z_n,\widetilde{Z}_{n+1})}
      \exp(-\beta {\mathcal W}^{C_n}_{n+1}),
    \end{equation}
   where
   \[
   \mathcal W^1_{n+1}=\sum_{k=0}^{K-1} V\left(\overline z_{k+1}(Z_n, \widetilde Z_{n+1}),\overline X^\perp_k\right) - V\left(\overline z_{k}(Z_n, \widetilde Z_{n+1}),\overline X^\perp_k \right),
   \]
   and
   \[
   \mathcal W^0_{n+1}=\sum_{k=0}^{K-1} V\left(\overline z_{k+1}(Z_n, \widetilde Z_{n+1}),\overline X^\perp_{k+1}\right) - V\left(\overline z_{k}(Z_n, \widetilde Z_{n+1}),\overline X^\perp_{k+1}\right).
   \]
   Both these works satisfy the Jarzynski equality~\eqref{eq:Jarz_consistency}.
\end{prop}
The proof of Proposition~\ref{eq:propchenrouxwork} can be found in \cref{proof:eq:propchenrouxwork}

\subsection{The symmetric path algorithm}
\label{sec:symm}

Let us now introduce a third algorithm which is symmetric for a reversible schedule in $z$ space, similarly to the second algorithm but without the need of sampling an additional Bernoulli random variable. This is reminiscent of what was proposed in \cite{athenesComputationChemicalPotential2002} in the setting of an alchemical transformation corresponding to a Widom insertion (see also \cref{rem:athenes}). The symmetry again enables us to show that, under an additional assumption, the algorithm can be seen as an example of the prototypical algorithm from \cref{sec:algo}.

At iteration $n$, let us consider the current state $X_n=(Z_n,X_n^\perp)$. Using the same ingredients as above, the iteration then proceeds in three steps:
\begin{itemize}
\item (Step 1) Propose a move in the collective variable space: $\widetilde{Z}_{n+1} \sim Q(Z_n,\widetilde{z}) \, {\rm d}\widetilde{z}$.
\item (Step 2)
  \begin{itemize}
  \item Introduce a time-discrete schedule $(\overline z_0=Z_n, \overline z_1, \ldots, \overline z_{K-1}, \overline z_K=\widetilde{Z}_{n+1})$ which is deterministically defined as a function of $(Z_n,\widetilde{Z}_{n+1})$. 
  \item Generate a sequence of states $(\overline X^\perp_{k})_{0 \le k \le K+1}$ as~$\overline X^\perp_0=X_n^\perp$ and~$\overline X^\perp_{k+1}$ sampled with respect to $P_{\overline z_{k}}(\overline X_k^\perp,x^\perp) \, {\rm d}x^\perp$ for~$k=0,\ldots,K$.
  \item Set $\widetilde X^\perp_{n+1}=\overline X^\perp_{K+1}$.
  \end{itemize}
\item (Step 3) Draw a random variable $U_{n+1}$ with uniform law over~$[0,1]$ and
\begin{itemize}
    \item If $\displaystyle U_{n+1} \le  \frac{\dps \exp(-\beta V(\widetilde Z_{n+1},\widetilde X_{n+1}^\perp)) Q(\widetilde{Z}_{n+1}, Z_n) \prod_{k=0}^{K} P_{\overline z_{k}(\widetilde Z_{n+1}, Z_n )}(\overline X^\perp_{K+1-k},\overline X^\perp_{K-k})}{\dps \exp(-\beta V(Z_n,X_n^\perp)) Q(Z_n,\widetilde{Z}_{n+1}) \prod_{k=0}^{K} P_{\overline z_{k} (Z_n,\widetilde Z_{n+1})}(\overline X^\perp_{k},\overline X^\perp_{k+1})}$,\\ accept the proposal: $(Z_{n+1},X^\perp_{n+1})=(\widetilde Z_{n+1},\widetilde X^\perp_{n+1})$.
    \item Otherwise, reject the proposal: $(Z_{n+1},X^\perp_{n+1})=(Z_{n}, X^\perp_{n})$.
\end{itemize}
\end{itemize}
We also prove the reversibility for this algorithm.

\begin{prop}
\label{eq:propstrangrev}
The Markov chain $(X_n)_{n \ge 0}$ generated by the above algorithm is reversible with respect to $\nu({\rm d}x) = \partition^{-1} \exp(-\beta V(z,x^\perp)) \, {\rm d}z \, {\rm d}x^\perp$.   
\end{prop}

The proof is given in \cref{proof:eq:propstrangrev}. As for the two previous algorithms, it does not require any specific assumption on the kernels and on the schedule in $z$.

Assume next that the schedule satisfies the symmetry property~\eqref{eq:rev_chemin_prime}.
Then, the last factor in the acceptance ratio of Step~3 can be rewritten as
\begin{align*}
    \prod_{k=0}^{K}\frac{ P_{\overline z_{k} (Z_n,\widetilde Z_{n+1})}(\overline X^\perp_{k+1}, \overline X^\perp_{k})}{P_{\overline z_{k} (Z_n,\widetilde Z_{n+1})}(\overline X^\perp_{k},\overline X^\perp_{k+1})},
\end{align*}
justifying why we call this algorithm the symmetric algorithm since the numerator involves the likelihood of the forward path followed backwards in time, for the successive forward kernels (this is illustrated on \cref{fig:SketchAlgos}: the red path is exactly the black path followed backwards).

Assuming furthermore that the equilibration kernels are reversible with respect to the conditional target measure (namely that~\eqref{eq:rev_second} holds), one can again rewrite the acceptance ratio in terms of the Jarzynski--Crooks works. This shows that this algorithm can be seen as one instance of our prototype algorithm from \cref{sec:algo}. 

\begin{prop}\label{eq:propStrangWorks}
   Under Assumptions~\eqref{eq:rev_chemin_prime} and~\eqref{eq:rev_second}, it holds
   \begin{equation}
     \label{eq:work2}
     \frac{\dps \exp(-\beta V(\widetilde Z_{n+1},\widetilde X_{n+1}^\perp)) Q(\widetilde{Z}_{n+1}, Z_n) \prod_{k=0}^{K} P_{\overline z_{k}(\widetilde Z_{n+1}, Z_n )}(\overline X^\perp_{K+1-k},\overline X^\perp_{K-k})}{\dps \exp(-\beta V(Z_n,X_n^\perp)) Q(Z_n,\widetilde{Z}_{n+1}) \prod_{k=0}^{K} P_{\overline z_{k} (Z_n,\widetilde Z_{n+1})}(\overline X^\perp_{k},\overline X^\perp_{k+1})}= \frac{Q(\widetilde{Z}_{n+1}, Z_n)}{Q( Z_n,\widetilde{Z}_{n+1})} \exp(-\beta {\mathcal W}^{2}_{n+1}),
   \end{equation}
   where
   \[
   \mathcal W^2_{n+1}=\sum_{k=1}^{K} V\left(\overline z_{k}(Z_n, \widetilde Z_{n+1}),\overline X^\perp_k\right) - V\left(\overline z_{k-1}(Z_n, \widetilde Z_{n+1}),\overline X^\perp_k\right).
   \]
   This work satisfies the Jarzynsky equality~\eqref{eq:Jarz_consistency}.
\end{prop}

The proof is provided in \cref{proof:eq:propStrangWorks}. The work takes the same form as the work $\mathcal W^0_{n+1}$ in \cref{eq:propchenrouxwork}, with the difference that the path in this symmetric algorithm comprises one additional equilibration step.

\begin{rem}
\label{rem:athenes}
The symmetric scheme we propose starts and ends with an equilibration step. We note that it is also possible to define a symmetric scheme starting and ending with an update of the CV as proposed in \cite{athenesComputationChemicalPotential2002}.
\end{rem}

\begin{rem}
  Assume here that the equilibration kernels are built using one step of the MALA algorithm (for a time step $\Delta t=T/K$) with target the conditional target measure, and that the time discrete schedules are given by $\overline z_k=z(k\Delta t)$, for some continuous-in-time schedule $(z(t))_{0 \le t \le T}$. Then one can check that the three acceptance ratios introduced in the three algorithms above converge, in the limit $\Delta t \to 0$, to the acceptance ratio of the continuous-in-time path algorithm of Section~\ref{sec:contalgo}. The discrete-in-time path algorithms can thus be seen as discretizations of the latter. However, as mentioned above, these discrete algorithms are also defined in a much more general setting, without referring to any underlying continuous-in-time path algorithm. It would be interesting to try to identify how to optimally choose these equilibration kernels. We intend to explore such optimization problems in future works.
\end{rem}

\section{Numerical Experiments}
\label{sec:numerics}
In all the implementations of the algorithms presented in \cref{sec:discretetime}, we use a simple linear interpolation schedule for the collective variable with
\[
\forall k\in\{0,\dots,K\}, \qquad \overline{z}_k = Z_n + \frac{ k}{K} (\widetilde{Z}_{n+1} - Z_n).
\]
Notice in particular that this schedule satisfies the symmetry property~\eqref{eq:rev_chemin_prime}. For the equilibration kernel $P_z(x^\perp, {\rm d}\widetilde{x}^\perp)$, we use discretized overdamped Langevin dynamics. More precisely, we use the unadjusted discretization (ULA) for all algorithms, and also test the reversible Metropolis-adjusted discretization (MALA) for the stochastically symmetric and the symmetric algorithm, since the acceptance ratios of these two simplify when using kernels reversible with respect to their target measure (see Propositions~\ref{eq:propchenrouxwork} and~\ref{eq:propStrangWorks}).  These variants give overall rise to five unbiased algorithms, which in shorthand will also refer to as Asym, CoinUla, SymUla, CoinMala, SymMala. We compare them in the following. 

To make ULA and MALA fully explicit, recall the Euler--Maruyama discretization of the overdamped Langevin dynamics with respect to the target measure $\nu$: At current state $x_k$, the new state is given by 
$$\widetilde{x}_{k+1} = x_k + \Delta t \nabla_x \log\nu(x_k) + \sqrt{2\Delta t}\xi_k,$$
for a real positive time step $\Delta t$ and where the random variables~$(\xi_k)_{k \geq 0}$ are independent and identically distributed according to a Gaussian distribution with zero mean and unit variance. For ULA, one simply has $x_{k+1} = \widetilde{x}_{k+1}$, whereas for MALA the new state is accepted with probability 
\begin{align}
\min \left(1, \frac{\nu(\widetilde{x}_{k+1})q(x_k|\widetilde{x}_{k+1})}{\nu(x_k)q(\widetilde{x}_{k+1}|x_k)}\right),
\end{align}
where $q$ is the transition density of one discretized step of the overdamped Langevin dynamics:
\begin{align}
q(\widetilde{x} |x) \propto \exp\left(-\frac{1}{4\Delta t} \left\|\widetilde{x}-x-\Delta t\nabla_x \log \nu(x)\right\|^2\right).
\end{align}

Since $\Delta t$ provides a time scale, it appears natural and consistent with the continuous time picture to prescribe for the corresponding CV schedule a velocity $v=|\widetilde{Z}_{n+1} - Z_n|/(K \Delta t)$, which means choosing the number of intermediate steps $K$ linearly with the distance in CV space. We choose this procedure for simplicity in all the experiments presented below.

\subsection{Moving in Collective Variable space}
As a first numerical experiment we demonstrate the advantage of utilizing the CV-guided moves in metastable systems by considering a bimodal target distribution in dimension $d=20$. The details of the model are described in \cref{subsec:fixedmodeswitch}.
For now it suffices to say that there is a `good' one-dimensional collective variable $z$, meaning that conditioning on~$z$, the target distributions on the remaining degrees of freedom~$x^\perp$ are unimodal. A short discussion on what constitutes good and bad collective variables can be found in \cref{sec:app_cv_choice_failure}.
As can be seen in \cref{fig:cvspaceillustration}, a naive approach with a MALA sampler on the full coordinate space fails to converge and never switches modes, even for very long chains. In contrast, a sampler that uses the symmetric algorithm with a Gaussian random walk as a proposal sampler $Q(Z_n, \widetilde{z})$ and MALA equilibration kernels switches modes regularly. This is in line with the general intuition that a random walk sampler is quickly bound to fail to sample a metastable system when the dimension grows. Meanwhile, the CV-guided samplers discussed in this paper have the potential to address metastable systems in relatively large dimensions.  

\begin{figure}
    \centering
    \includegraphics{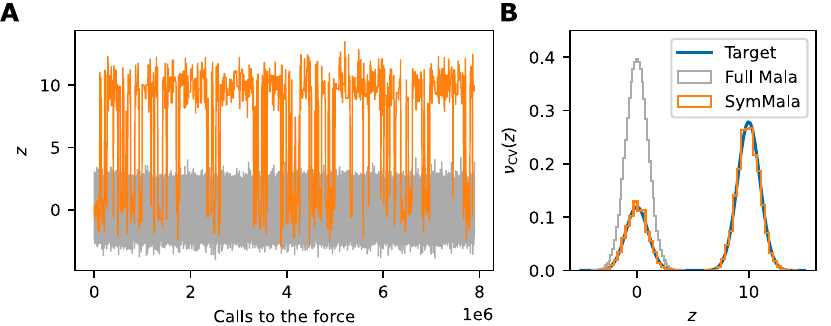}
    \caption{Sampling from a bimodal distribution in 20 dimensions with a naive MALA algorithm fails to switch modes. In contrast, SymMala with a random walk in the 1D CV space frequently switches modes, so that the mode ratios are accurately represented in the samples. The example is the same as in \cref{fig:dumbell_transition}, see \cref{subsec:fixedmodeswitch} for details. Both algorithms were tuned such that the $z$ moves get accepted $\sim$20\% of the time, corresponding, for SymMala, to a step size of $\Delta Z=6$ for the random walk (with $v=0.0075$, $\Delta t=0.4$) and, for the full MALA, to $\Delta t=0.65$. For a fair comparison, trajectories are shown against the number of calls to the force. For SymMala, the total number of $\sim 8$ million steps is the result of 5,000 steps in CV space.}
    \label{fig:cvspaceillustration}
\end{figure}

\subsection{Discretization of the dynamics}
\label{subsec:discreteizationexample}

To illustrate the point already made in \cref{sec:discretetime} that the discretization of the continuous-in-time algorithm from \cref{sec:contalgo} has to be handled with care, we consider the two-dimensional Gaussian target measure depicted in \cref{fig:simple2D_bias}. It has a probability density of $\nu(z,x^\perp)=(2\pi)^{-1}\exp(-0.5z^2 - 0.5(x^\perp-z)^2)$. We generate MCMC chains with the five different algorithms described above, which were proven to be unbiased in \cref{sec:discretetime}. We also consider the ULA and MALA versions of the stochastically symmetric algorithm where the coin flip necessary to determine whether the path starts with a CV or an $x^\perp$ update is omitted, so that the forward path always starts with a CV update. Finally, we add the symmetric algorithm where the non-reversible ULA kernel is used, but the acceptance ratio is based on the work that is appropriate for the reversible MALA kernel. In \cref{fig:simple2D_bias}, we compare the empirical marginal distributions produced by the different algorithms. As expected, for a large velocity, corresponding to CV schedules with just a handful of intermediate steps, a noticeable bias shows up in the last three algorithms. Going closer to the continuous limit by using a smaller velocity and smaller time step (corresponding to a larger number of intermediate steps, here by a factor of 65), the bias becomes unnoticeable in this simple example. 

\begin{figure}
    \centering
    \includegraphics{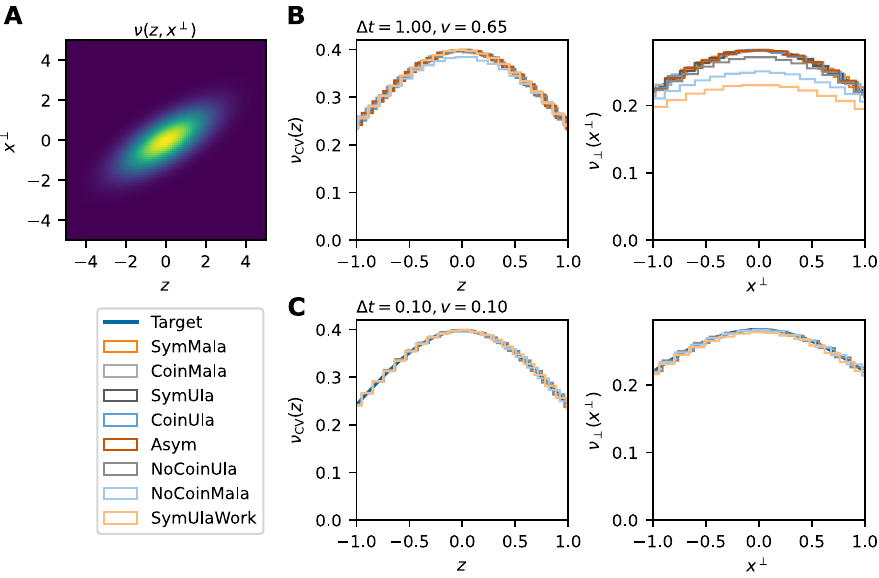}
    \caption{Discretizing the continuous-time schedule from \cref{sec:contalgo} requires special care to avoid introducing a bias, as illustrated here on a simple unimodal 2D example. The target probability density is visualized in Panel~A and described more precisely in \cref{subsec:discreteizationexample}. Panel~B: A single chain was run using 8 different versions of the algorithm for $10^7$ steps until the histograms converged. As expected, the first five algorithms (for which we proved the unbiasedness) do not show a bias and match the variance of the target, whereas the last three do not. Panel~C: Repeating the same experiments with smaller~$\Delta t$ and~$v$ towards the continuous limit, the bias is no longer visible.}
    \label{fig:simple2D_bias}
\end{figure}

\begin{figure}
    \centering    \includegraphics{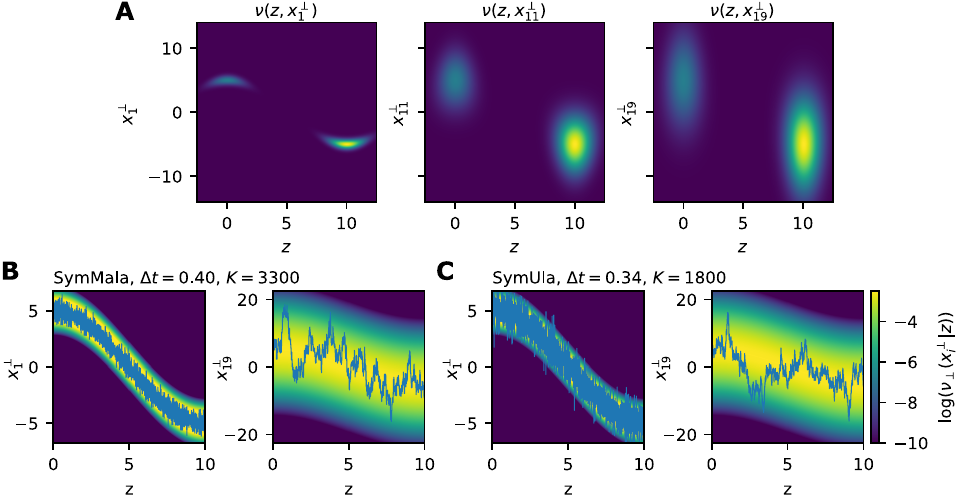}
    \caption{Panel~A: Illustration of the 20D toy model made precise in \cref{subsec:fixedmodeswitch} for a selection of coordinates $x_i^\perp$. Panels~B and~C: Representative mode switch trajectories (accepted with probability one in Step~3 of the algorithms) for two versions of the symmetric algorithm (with optimized parameters), shown for the two coordinates with smallest and largest variance (stiffest and least stiff degrees of freedom).}
    \label{fig:dumbell_transition}
\end{figure}
\begin{figure}
    \centering
\includegraphics{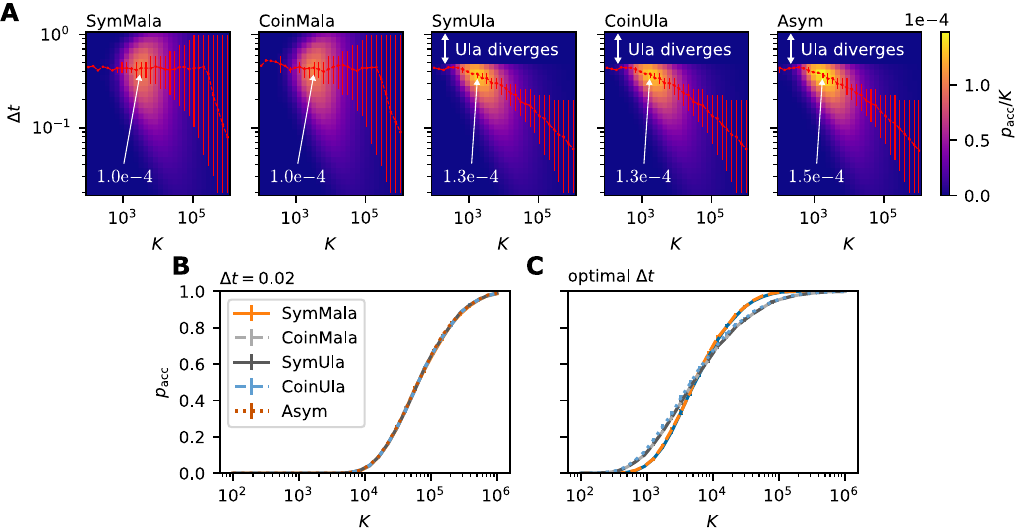}
    \caption{Targeted mode jump from $\overline z_0=0$ to $\overline z_K=b$, averaged over 10,000 realizations with $x^\perp_0 \sim p(x_0^\perp|z_0)$ for the 20D toy model described in \cref{subsec:fixedmodeswitch}. Panel~A: The red line indicates the optimal step size $\Delta t$ as a function of the number of intermediate steps $K$, averaged over 10,000 jumps. Bars indicate the range of $\Delta t$ where acceptance is at most 5\% worse than the optimum. The values leading to the optimal performance are explicitly stated and marked by an arrow. Panel~B: Acceptance rate as a function of~$K$ for fixed (small)~$\Delta t$. Panel~C: Acceptance rate as a function of~$K$ where $\Delta t$ is chosen optimally for each $K$.}
    \label{fig:jumpscatterplot}
\end{figure}

\subsection{Efficiency comparison on a fixed mode switch}
\label{subsec:fixedmodeswitch}

We now compare the performance of the five unbiased algorithms for a toy model in $d=20$ dimensions. It has a designated one-dimensional collective variable $z$ and $d-1$ additional coordinates $(x^\perp_i)_{1\leq i\leq d-1}$. The distribution of~$z$ is given by a Gaussian mixture and the conditional distribution of the remaining coordinates is a Gaussian with a mean depending on $z$. More precisely,
\begin{align}
    \nu(\mathrm{d}z,\mathrm{d}x^\perp) &= \nu_{\rm CV}(\mathrm{d}z) \nu_\perp(\mathrm{d}x^\perp|z), \notag\\
    \nu_{\rm CV}(\mathrm{d}z) & = \frac{1}{\sqrt{2\pi}} \left(w \, \mathrm{e}^{-z^2/2} + (1-w)\, \mathrm{e}^{-(z-b)^2/2}\right)\mathrm{d}z,
    \label{eq:z_marginal}\\
    \nu_\perp(\mathrm{d}x^\perp|z) & = (2\pi )^{-(d-1)/2} \mathrm{det}(\Sigma)^{-1/2} \exp \left(-{\frac {1}{2}}\left({x^\perp }-\mu(z)\right)^{\top}\Sigma ^{-1}\left({x^\perp }-\mu(z)\right)\right) \, \mathrm{d}x^\perp.\notag
\end{align}
Here, $\mu(z)=\cos(z/(b\pi)) \cdot b/2$ is the conditional mean function and $\Sigma$ a diagonal covariance matrix with ${{\Sigma }}_{ii}=\sigma_i^2$. The weight of the first mode is~$w=0.3$, while the distance between the modes is~$b=10$. The values of $\sigma_i$ were chosen from an evenly spaced grid in the interval $[0.5, 5]$. The target measure is illustrated in \cref{fig:dumbell_transition}, where a selection of marginalized measures over two coordinates is shown. Here and in some instances in the following, we use the notation $\nu(\mathrm{d}z,\mathrm{d}x^\perp_i)$ or $\nu_\perp(\mathrm{d}x^\perp_i|z)$ for a specific coordinate $x^\perp_i$, by which we always indicate the distribution that has been marginalized over the remaining degrees of freedom $\{x^\perp_j\}_{i\neq j}$.

In a first experiment, we target the specific mode transition from $z_0=0$ to $z_1=b$, where the initial orthogonal coordinates~$x^\perp$ are sampled from the conditional target distribution. A good performance criterion of the algorithm is the acceptance rate of this jump divided by the number of intermediate steps $K$, since the latter equals the number of calls to the force and hereby the cost. In panel A of \cref{fig:jumpscatterplot}, this normalized acceptance is shown for a grid of different values of~$\Delta t$ and~$K$ for the five unbiased algorithms. All the ULA-based algorithms perform equally well and have the same optimal $(\Delta t,K)$ pair. The same applies to the MALA-based algorithms. In this instance, ULA-based algorithms reached slightly higher values of normalized acceptance. 
Regarding optimal performance, the MALA-based algorithms are very robust in that they have a value of optimal $\Delta t$ which is almost independent of $K$ in the range of values considered here. For large values of $K$, the performance as a function of $\Delta t$ is very broadly peaked so that $\Delta t$ can be chosen from a large window with similar performance (see red vertical bars). In particular, this allows for an optimization of $\Delta t$ at small $K$ in an initial step, followed by optimization of $K$. In contrast, the optimal $\Delta t$ value decreases as a function of $K$ for the ULA-based algorithms and the peak of optimal performance (for fixed $K$) is less broad. Notably, these algorithms fail completely for a choice of $\Delta t$ that is too large, as the ULA dynamics will diverge. Especially for smaller values of~$K$, the optimal $\Delta t$ is close to this point and thus on the verge of stability.

For small values of $\Delta t$ close to the continuous limit, the performance of all five algorithms is on par, as can be seen from Panel~B of \cref{fig:jumpscatterplot}. In Panel~C, the acceptance rate is shown at fixed $K$ values for the respective optimal values of~$\Delta t$. For small $K$, the ULA-based algorithms perform slightly better than the MALA-based algorithms, even though the converse is true beyond $K\approx 10^4$.

Typical transition trajectories are shown in \cref{fig:dumbell_transition} for the symmetric algorithm, with MALA (Panel~B) and ULA (Panel~C) equilibration kernels, for their respective optimal values for~$(\Delta t,K)$. Since $\Delta t$ defines the step size in the Langevin dynamics, one would expect its optimal value to be constrained by the stiffest degree of freedom, i.e. the one with the smallest variance, and this appears indeed to be the case here.

\subsection{Efficiency comparison for full MCMC trajectories}

We also tested the performance of the five different samplers in generating full MCMC trajectories for the model considered in \cref{subsec:fixedmodeswitch}. We draw proposals in CV space independently with a proposal density that does not depend on the current value, $Q(z,\widetilde{z})=Q(\widetilde{z})$. The density of this proposal is of the same form as $\nu_{\rm CV}$ defined in \cref{eq:z_marginal}, but different from the target distribution, as both mixture components are proposed with the equal weight $w=0.5$. 

As mentioned before, the number of intermediate steps for each proposed move is chosen proportionally to its distance as a simple, consistent strategy, namely $K = |\widetilde{Z}_{n+1} - Z_n|/(v \Delta t)$ for a given velocity~$v>0$. As a performance criterion, we consider the inverse of the average number of calls to the force between observing mode switch. In \cref{fig:modeswitchcost}B, this performance is shown for the five algorithms for different pairs of $(v,\Delta t)$. The $x$-axis represents the number of steps for the jump from one mode center to the other to make the plots comparable with \cref{fig:jumpscatterplot}A and indeed the pictures are almost identical. The ULA and MALA based algorithms perform equally well among each other and there is a slight advantage for ULA, although the same issue as in \cref{subsec:fixedmodeswitch} applies here: if $\Delta t$ is chosen too large for ULA, the ULA-based algorithms become unstable.

As illustrated in Panel~A of \cref{fig:modeswitchcost} where marginal distributions are displayed, the samplers manage to accurately sample from the target distribution with correct weights and correct for the mismatch between the proposal and the target distribution. 

\begin{figure}
    \centering
    \includegraphics{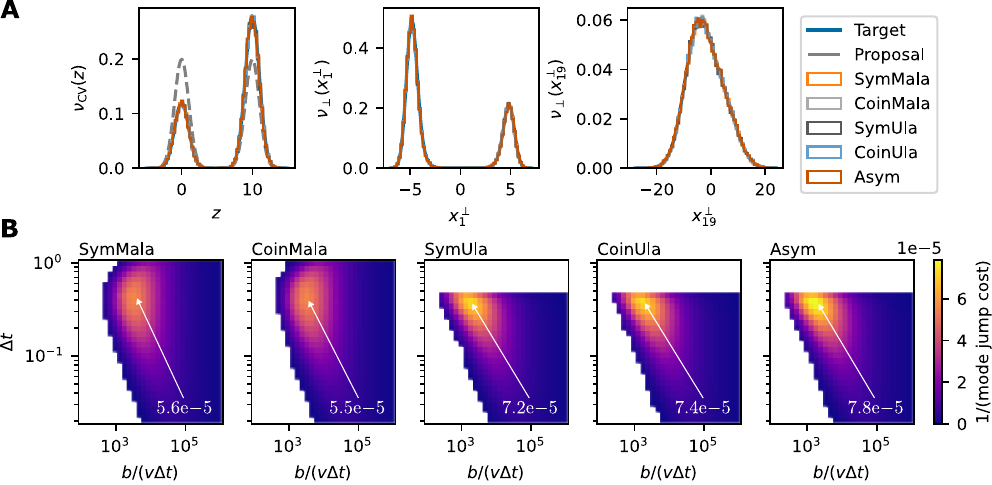}
    \caption{Comparison of the algorithms' performances for full MCMC chains on the 20D target from \cref{subsec:fixedmodeswitch}, always initialized at $z=0$ and $x^\perp$ given by the conditional mean function $\mu(z=0)$. Panel~A: Marginal distributions with optimal $\Delta t$ and $K$ for each algorithm. Histograms are averaged over 8 chains of length 20,000. Panel~B: Inverse cost of mode switch  (average number of MALA/ULA steps between mode switches) as $\Delta t$ and $v$ vary. For comparison with \cref{fig:jumpscatterplot}A, the variation of $v$ along the $x$-axis is converted to  $b/(v \Delta t)$, which represents the number of steps for a jump from one mode to the other for the marginal distribution~\cref{eq:z_marginal}.
   White space indicates that at least one of eight chains of length 10,000 did not change mode at all. The values leading to the optimal performance are explicitly stated and marked by an arrow.}
    \label{fig:modeswitchcost}
\end{figure}

\subsection[Phi4 Model]{$\phi^4$ Model}\label{sec:phi4}
Finally, we demonstrate the CV-guided MCMCs on a field model from statistical physics, the 1D-$\phi^4$ model. It consists of $N$ continuous degrees of freedom $\{\phi_i\}_{i=1}^N$ on a lattice and can be seen as a continuous version of the Ising model modeling ferromagnetism in solids. The target distribution is the Boltzmann distribution $\nu({\rm d}\phi)$ from~\eqref{eq:boltzmannmeasure} with the energy (using a simplified version of \cite{gabrieAdaptiveMonteCarlo2022})
\begin{align}
   \beta V(\phi) = \frac{a\beta N}{2} \sum_{i=1}^{N+1}(\phi_i - \phi_{i-1})^2 + \frac{\beta}{4 a N}\sum_{i=1}^N (1-\phi_i^2)^2,
    \label{eq:phi4}
\end{align}
with Dirichlet boundary conditions $\phi_0=\phi_{N+1} = 0$. For sufficiently low temperatures~$1/\beta$, this system has two well-separated modes characterized by positive and negative values of the magnetization~$M=\sum_{i=1}^N \phi_i / N$. We place ourselves in this bimodal phase by setting~$\beta=10$ and~$a=0.1$. In many practical settings, one will be interested in systems with an external field $h$ adding a term $h\sum_{i=1}^N \phi_i$ to the potential $V(\phi)$.
The special case of $h=0$ that we consider, however, is particularly insightful since it has sign-inversion symmetry, $\nu(\phi)=\nu(-\phi)$, so that the weights of the two modes are known to be exactly equal. This offers an easily accessible checkmark to judge the quality of samples obtained from any MCMC algorithm.

The magnetization is the obvious choice of collective variable for the algorithms presented here. A natural way to obtain it is using a Haar wavelet decomposition of the fields (see \cref{sec:wavelet}). We applied the symmetric algorithm with MALA equilibration kernels (SymMala) on the coordinates in the wavelet basis with $M$ as the collective variable. 

In \cref{fig:phi4}.A, accepted transitions from one mode to the other are shown by displaying their state at intermediary values of the magnetization $M$ along the Jarzynski--Crooks path. Note here that the intermediate target measures $\nu_\perp(x^\perp|M)$ are bimodal, so that two possible types of paths between the two modes coexist. Physically speaking, a domain wall between areas of positive and negative orientation is created and they can be arranged either as $(-|+)$ or as $(+|-)$, colored respectively in red and blue.
However, the magnetization is a good index of the two possible paths so that the bimodality does not pose a fundamental problem in effectuating transitions.

As a basic sanity check of the algorithm, we also simulate full MCMC trajectories using SymMala. As a `ground truth' for comparison, we run a simple MALA trajectory (staying within one mode) over 200,000 time steps with $\Delta t=0.003$ and symmetrize the resulting samples by hand, flipping the sign of each configuration with probability $0.5$. 
For the SymMala algorithm, we propose configurations in CV space in an independent fashion so that the density does not depend on the current value, $Q(z,\widetilde{z})=Q(\widetilde{z})$. We use a Gaussian mixture model with Gaussians centered at $\pm M^*$ with variance $\sigma^2$. Both parameters are empirically estimated using the simple MALA trajectory within one mode and we obtain $M^* \approx 0.74$ and $\sigma^2\approx 0.01$. Since we aim to show that the algorithm accurately corrects for any mismatch between proposal and target distribution, we choose the weights of the Gaussian mixture as $0.3$ and $0.7$. The resulting distributions are shown in Panel~B of \cref{fig:phi4} with the marginal distributions of the magnetization and the mid-point amplitude $\phi_{N/2}$, and indeed the obtained samples accurately attribute equal weight to the two modes.
\begin{figure}
    \centering
    \includegraphics{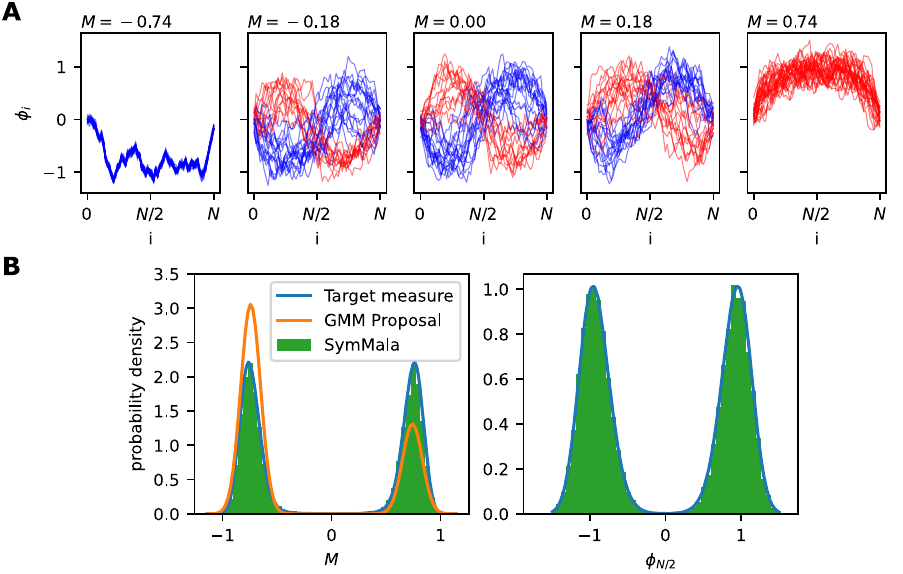}  
    \caption{1D $\phi^4$ model from \cref{eq:phi4} with $N=64$, $\beta=10$, $a = 0.1$. The magnetization $M$ is used as a collective variable. Panel~A: Transition trajectories between the mode centers with the SymMala algorithm with $K=20,000$, $\Delta t = 0.001$. A total of 1,000 proposal trajectories was simulated, all with equal initial field configuration $\phi$, which is shown in the leftmost plot. Of these constructed paths, the 30 with the highest acceptance probability are shown here.
    The coloring is chosen based on the sign of the average value in the left half, $\sum_{i=1}^{N/2}\phi_i / (N/2)$. Panel~B: Marginal distributions for SymMala with the same $\Delta t$ and $v=0.074$ (corresponding to the same number $K$ for a mode switch). The CV sampler is an iid Gaussian mixture model (GMM) with weights~0.3 and~0.7. The resulting histogram is averaged over 8 equally initialized chains with 10,000 steps. The target distribution shown for comparison was generated using samples from a simple MALA sampler symmetrized by hand.}
    \label{fig:phi4}
\end{figure}

\section{Conclusion and perspectives}

The present work investigates how to leverage a proposal Markov kernel in a CV space to sample a possibly metastable Boltzmann distribution. The general framework relies on Jarzynski--Crooks moves which combine a scheduled path in CV space and equilibration kernels for the remaining degrees of freedom to form a forward transition path towards a newly proposed CV value associated with a work value. We prove the reversibility of the resulting Markov chain by considering first continuous-in-time Jarzynski--Crooks moves. Moving on to discrete-in-time paths, we show that various backward paths can be considered in the acceptance criteria and describe the connections with previous work \cite{nilmeier2011nonequilibrium,chen2015generalized,athenesComputationChemicalPotential2002,neal2005taking}. 
We prove the reversibility of three discrete algorithms. 

Through numerical experiments, we evidence that lack of care in defining the relationship between the forward and backward paths can result in biased algorithms. Comparing different versions of the discrete algorithms on a high-dimensional metastable system, we find that all versions perform comparably. A slight advantage is found for using MALA as an equilibration kernel, as the performances of the algorithms are less sensitive to the choice of the discretization time step. We also demonstrate the ability of the Jarzynski--Crooks moves to power mode-to-mode jump on a bi-modal $\phi^4$ field theory where two channels of transitions co-exist.  

Throughout this work, we assumed that a good collective variable along with a good proposal kernel in CV space were available. 
Note that in general there will be a trade-off in the choice of dimensionality of the collective variable. For CV spaces of very small dimension, it is easy to build a good CV sampler (Step~1 of the presented algorithms could rely on a Gaussian random walk for instance). However, it will be hard to equilibrate the remaining coordinates and thus to build a path toward a new CV value with good acceptance. Conversely, for CV spaces of more than a few dimensions, the equilibration of the remaining coordinates $x^\perp$ conditioned on the CV will be easy, but building a good sampler in CV space will be hard.
In this case, one would need a more sophisticated strategy to explore the various metastable basins. As already mentioned, this is the focus of the parallel work \cite{samuel}, employing the framework of the present work with a normalizing flow to propose moves in CV space with an adaptive strategy allowing for the learning of the generative model along the sampling. It would also be interesting to explore other CV sampling techniques, based for example on free energy adaptive biasing, temperature accelerated Molecular Dynamics, multiproposal versions of the Metropolis Hastings algorithm, etc, see~\cite{athenesFreeEnergyReconstruction2010} for results in these directions.

Another question brought up by the present work is the optimization of the computational cost of the Jarzynski--Crooks paths.
We considered a linear interpolation schedule in CV space and Langevin-based equilibration kernels for the remaining degrees of freedom. These choices, made for simplicity, are not necessary for reversibility and others could be relevant. 
This is an interesting direction that we leave for future study. 

\paragraph{Acknowledgements.} We would like to thank Alessandro Laio (SISSA) and Samuel Tamagnone (SISSA) for insightful discussions. We would like to thank Andi Wang (University of Warwick) and Guanyang Wang (Rutgers University) for enlightening discussions concerning Section~\ref{sec:prototype_proof}. We would like also to thank the anonymous referee as well as Arnaud Doucet for pointing out the reference~\cite{neal2005taking} to us. The works of T.L. and G.S. are partially funded by the European Research Council (ERC) under the European Union’s Horizon 2020 research and innovation programme (project EMC2, grant agreement No 810367), and from the Agence Nationale de la Recherche
through the grants ANR-19-CE40-0010-01 (QuAMProcs) and ANR-21-CE40-0006 (SINEQ). The works of C.S. and M.G. are supported by Hi! Paris.

A minimal running implementation of the algorithms presented in this work is publicly available at\\ \url{https://github.com/cschoenle/mc_jarz_sampling}.

\bibliographystyle{alpha}
\bibliography{sample}

\appendix

\section{Proofs}

\subsection{Proof of Theorem \ref{th:reversibility}}
\label{proof:th:reversibility}

Let us start by introducing some notation. When needed, we write $\mathcal{W}^{(z,\widetilde z)}$ instead of $\mathcal{W}$, the superscript indicating explicitly the dependence of the law of $\mathcal{W}_{n+1}$ on $(Z_n,\widetilde Z_{n+1})=(z,\widetilde z)$. Likewise, we denote below when needed $(\overline X^{(z,\widetilde z)}_t)$ and $z^{(z,\widetilde z)}(t)$ to explicitly indicate the dependence on~$(Z_n,\widetilde Z_{n+1})=(z,\widetilde z)$ of the associated process and schedule.

Assume that $X_n=(Z_n, X^\perp_n)$ is distributed according to $\nu$. The aim is to prove that $(X_n,X_{n+1})$ has the same law as $(X_{n+1},X_n)$. In order to study the law of $(X_n,X_{n+1})$, let us consider a bounded measurable function~$\varphi:\R^d \times \R^d \to \R$.
One has:
\begin{align*}
\E(\varphi(X_n,X_{n+1}))
&=\E\left(\varphi(X_n,\widetilde{X}_{n+1}) \one_{U_{n+1}\le \exp(-\beta \mathcal{W}_{n+1})\frac{Q(\widetilde{Z}_{n+1}, Z_n)}{Q(Z_n,\widetilde{Z}_{n+1})}}\right)\\
&\quad + \E\left(\varphi(X_{n},X_n)\one_{U_{n+1}> \exp(-\beta \mathcal{W}_{n+1})\frac{Q(\widetilde{Z}_{n+1}, Z_n)}{Q(Z_n,\widetilde{Z}_{n+1})}}\right) \\
&= \E\left(\varphi(X_n,\widetilde{X}_{n+1}) 1 \wedge \left[ \exp(-\beta \mathcal{W}_{n+1})\frac{Q(\widetilde{Z}_{n+1}, Z_n)}{Q(Z_n,\widetilde{Z}_{n+1})}\right]\right)\\
&\quad + \E\left(\varphi(X_n, X_{n}) \left[ 1-  1 \wedge \left(\exp(-\beta \mathcal{W}_{n+1})\frac{Q(\widetilde{Z}_{n+1}, Z_n)}{Q(Z_n,\widetilde{Z}_{n+1})}\right)\right]\right).
\end{align*}
By conditioning on $(Z_n,\widetilde{Z}_{n+1})$, one gets
\begin{align}
  &\E\left(\varphi(X_n,\widetilde{X}_{n+1}) 1 \wedge \left[ \exp(-\beta \mathcal{W}_{n+1})\frac{Q(\widetilde{Z}_{n+1}, Z_n)}{Q(Z_n,\widetilde{Z}_{n+1})}\right]\right) \label{eq:term_acceptance} \\
  &=\int_{z, \widetilde z} \E\left[\varphi(X_n,\widetilde{X}_{n+1}) 1 \wedge \left( \exp(-\beta \mathcal{W}_{n+1})\frac{Q(\widetilde{Z}_{n+1}, Z_n)}{Q(Z_n,\widetilde{Z}_{n+1})}\right)\middle| (Z_n,\widetilde Z_{n+1})=(z,\widetilde z) \right] \exp(-\beta F(z)) Q(z,\widetilde z) \, {\rm d}z \, {\rm d}\widetilde z \notag\\
  &=\int_{z, \widetilde z} \E\left(\varphi(X_n,\widetilde{X}_{n+1}) \left[ Q(Z_n,\widetilde{Z}_{n+1}) \wedge \left( \exp(-\beta \mathcal{W}_{n+1}) Q(\widetilde{Z}_{n+1},Z_n)\right)\right]\middle| (Z_n,\widetilde Z_{n+1})=(z,\widetilde z) \right) \exp(-\beta F(z))  \, {\rm d}z \, {\rm d}\widetilde z.\notag
\end{align}
Notice that in the previous expectation, $X_n=(Z_n,X_n^\perp)$, with $X_n^\perp$  distributed according to the conditional probability measure~$\nu_{\perp}({\rm d} x^\perp|Z_n)$. 

Using Proposition~\ref{prop:Jarz-Crooks}, one gets, using the forward process $(\overline X_t)_{0 \le t \le T}$ defined by~\eqref{eq:forward_process}--\eqref{eq:Jarz_IC_eq}, as well as the backward process $(\overline X^{\rm b}_t)_{0 \le t \le T}$ defined by~\eqref{eq:backward_process}-\eqref{eq:eq:Jarz_ICb_eq}:
\begin{align*}
  &\E\left(\varphi(X_n,\widetilde{X}_{n+1}) \left[ Q(Z_n,\widetilde{Z}_{n+1}) \wedge \left( \exp(-\beta \mathcal{W}_{n+1}) Q(\widetilde{Z}_{n+1},Z_n)\right)\right]\middle| (Z_n,\widetilde Z_{n+1})=(z,\widetilde z) \right)\\
  &=\E\left(\varphi(\overline X_0,\overline X_T) \left[ Q(z,\widetilde z) \wedge \left( \exp\left(-\beta \int_0^T \left\langle \nabla_z V(\overline X_t), z'(t)\right\rangle \, {\rm d}t\right) Q(\widetilde z, z)\right)\right]\right)\\
  &=\E\left(\varphi(\overline X_0,\overline X_T) \left[\left( \exp\left(\beta \int_0^T    \left\langle \nabla_z V(\overline X_t), z'(t)\right\rangle \, {\rm d}t\right)  Q(z,\widetilde z) \right)\wedge  Q(\widetilde z, z)\right] \exp\left(-\beta \int_0^T \left\langle \nabla_z V(\overline X_t), z'(t)\right\rangle \, {\rm d}t\right) \right)\\
  &=\E\left(\varphi(\overline X^{\rm b}_T,\overline X^{\rm b}_0) \left[ \left(\exp\left(\beta \int_0^T \left\langle \nabla_z V(\overline X^{\rm b}_{T-t}), z'(t)\right\rangle\, {\rm d}t\right) Q(z,\widetilde z) \right)\wedge   Q(\widetilde z, z)\right]\right) \exp(-\beta (F(\widetilde z)-F(z)))\\
  &=\E\left(\varphi(\overline X^{\rm b}_T,\overline X^{\rm b}_0) \left[ \left(\exp\left(\beta \int_0^T \left\langle \nabla_z V(\overline X^{\rm b}_{t}), z'(T-t)\right\rangle\, {\rm d}t\right) Q(z,\widetilde z) \right)\wedge   Q(\widetilde z, z)\right]\right) \exp(-\beta (F(\widetilde z)-F(z))).
\end{align*}
Under the assumption~\eqref{eq:hyp_rev_z}, the backward process~$(\overline X^{\rm b}_t)_{0 \le t \le T}$ in~\eqref{eq:backward_process} for the schedule associated with the end points~$(z,\widetilde z)$ has the same law as the forward process~$(\overline X_t)_{0 \le t \le T}$ in~\eqref{eq:forward_process} for the schedule associated with the end points~$(\widetilde z,z)$. This shows that (notice that $(z^{(z,\widetilde z)})'(T-t) =- (z^{(\widetilde z,z)})'(t)$ by deriving~\eqref{eq:hyp_rev_z})
\begin{align*}
  &    \E\left(\varphi(\overline X^{\rm b}_T,\overline X^{\rm b}_0) \left[\left( \exp\left(\beta \int_0^T \left\langle \nabla_z V(\overline X^{\rm b}_{t}), z'(T-t)\right\rangle\, {\rm d}t\right) Q(z,\widetilde z) \right) \wedge   Q(\widetilde z, z)\right]\right)\\
  &=\E\left(\varphi(\overline X^{\widetilde z,z}_T,\overline X^{\widetilde z,z}_0) \left[\left( \exp\left(- \beta \int_0^T \left\langle \nabla_z V(\overline X^{\widetilde z,z}_{t}), (z^{\widetilde z,z})'(t)\right\rangle\, {\rm d}t\right) Q(z,\widetilde z) \right)\wedge   Q(\widetilde z, z)\right]\right)\\
  &=\E\left(\varphi(\overline X^{\widetilde z,z}_T,\overline X^{\widetilde z,z}_0) \left[ \left( \exp\left(- \beta \mathcal W^{(\widetilde z,z)}\right) Q(z,\widetilde z) \right)\wedge  Q(\widetilde z, z)\right]\right).
\end{align*}
Therefore, one obtains
\begin{align*}
    &\E\left(\varphi(X_n,\widetilde{X}_{n+1}) \left[ Q(Z_n,\widetilde{Z}_{n+1}) \wedge \left( \exp(-\beta \mathcal{W}_{n+1}) Q(\widetilde{Z}_{n+1},Z_n)\right)\right]\middle| (Z_n,\widetilde Z_{n+1})=(z,\widetilde z) \right)\\
&= \E\left(\varphi(\widetilde{X}_{n+1},X_n) \left[ \left(\exp(-\beta \mathcal{W}^{(Z_n,\widetilde Z_{n+1})}) Q(\widetilde{Z}_{n+1},Z_n) \right) \wedge   Q(Z_n,\widetilde{Z}_{n+1})\right] \middle| (Z_n,\widetilde Z_{n+1})=(\widetilde z,z) \right)\\
& \quad \times \exp(-\beta [F(\widetilde z)-F(z)]).
\end{align*}
Using this equality in~\eqref{eq:term_acceptance}, one gets:
    \begin{align*}
      &\E\left(\varphi(X_n,\widetilde{X}_{n+1}) \left[ 1 \wedge \left( \exp(-\beta \mathcal{W}_{n+1})\frac{Q(\widetilde{Z}_{n+1}, Z_n)}{Q(Z_n,\widetilde{Z}_{n+1})}\right) \right] \right)  \\
&=\int_{z, \widetilde z} \E\left(\varphi(\widetilde{X}_{n+1},X_n) \left[ \left(\exp(-\beta \mathcal{W}^{(Z_n,\widetilde Z_{n+1})} )Q(\widetilde{Z}_{n+1},Z_n) \right) \wedge   Q(Z_n,\widetilde{Z}_{n+1})\right] \middle| (Z_n,\widetilde Z_{n+1})=(\widetilde z,z) \right) \exp(-\beta F(\widetilde z))  \, {\rm d}z \, {\rm d} \widetilde z\\
&=\int_{z, \widetilde z} \E\left(\varphi(\widetilde{X}_{n+1},X_n) \left[ \left(\exp(-\beta \mathcal{W}^{(Z_n,\widetilde Z_{n+1})} )\frac{Q(\widetilde{Z}_{n+1},Z_n)}{Q(Z_n,\widetilde{Z}_{n+1})} \right) \wedge 1   \right]\middle| (Z_n,\widetilde Z_{n+1})=(\widetilde z,z) \right) \exp(-\beta F(\widetilde z)) Q(\widetilde z,z) \, {\rm d}z \, {\rm d} \widetilde z\\
&=\E\left(\varphi(\widetilde{X}_{n+1},X_n) \left[ 1 \wedge \left(\exp(-\beta \mathcal{W}_{n+1} )\frac{Q(\widetilde{Z}_{n+1},Z_n)}{Q(Z_n,\widetilde{Z}_{n+1})} \right) \right]\right).
\end{align*}
Finally,
    \begin{align*}
\E(\varphi(X_n,X_{n+1}))
&= \E\left(\varphi(\widetilde{X}_{n+1},X_n) 1 \wedge \left( \exp(-\beta \mathcal{W}_{n+1}) \frac{Q(\widetilde{Z}_{n+1}, Z_n)}{Q(Z_n,\widetilde{Z}_{n+1})}\right)\right)\\
&\quad + \E\left(\varphi(X_n, X_{n}) \left( 1-  1 \wedge \left(\exp(-\beta \mathcal{W}_{n+1})\frac{Q(\widetilde{Z}_{n+1}, Z_n)}{Q(Z_n,\widetilde{Z}_{n+1})}\right)\right)\right)\\
&=\E(\varphi(X_{n+1},X_n)),
\end{align*}
which concludes the proof.

\subsection{Proof of Proposition~\ref{eq:PropNilmeier}}
\label{proof:eq:PropNilmeier}

Assume that $X_n$ is distributed according to $\nu$, and consider the law of~$(X_n,X_{n+1})$, where~$X_{n+1}$ is obtained from~$X_n$ by the algorithm of Section~\ref{sec:Nilmeier}. Recall the notation $X_{n+1}=(Z_{n+1},X^\perp_{n+1})$ and $\widetilde X_{n+1}=(\widetilde Z_{n+1},\widetilde X^\perp_{n+1})$. For any bounded measurable function $\varphi:\R^d \times \R^d \to \R$, one has
\begin{align*}
&  \E(\varphi(X_n,X_{n+1}))\\
&= \E\left[\varphi(X_n,\widetilde X_{n+1}) 1 \wedge \left( \frac{\exp(-\beta V(\widetilde Z_{n+1},\widetilde X_{n+1}^\perp)) Q(\widetilde{Z}_{n+1}, Z_n) \prod_{k=1}^K P_{\overline z_k(\widetilde Z_{n+1},Z_n)}(\overline X^\perp_{K-k+1},\overline X^\perp_{K-k})}{\exp(-\beta V(Z_n,X_n^\perp)) Q(Z_n,\widetilde{Z}_{n+1}) \prod_{k=1}^K P_{\overline z_k(Z_n,\widetilde Z_{n+1})}(\overline X^\perp_{k-1},\overline X^\perp_{k})}\right) \right]\\
&  + \E\left[\varphi(X_n, X_{n}) \left(1- 1 \wedge \left( \frac{\exp(-\beta V(\widetilde Z_{n+1},\widetilde X_{n+1}^\perp)) Q(\widetilde{Z}_{n+1}, Z_n) \prod_{k=1}^K P_{\overline z_k(\widetilde Z_{n+1},Z_n)}(\overline X^\perp_{K-k+1},\overline X^\perp_{K-k})}{\exp(-\beta V(Z_n,X_n^\perp)) Q(Z_n,\widetilde{Z}_{n+1}) \prod_{k=1}^K P_{\overline z_k(Z_n,\widetilde Z_{n+1})}(\overline X^\perp_{k-1},\overline X^\perp_{k})}\right) \right)\right] \\
&=  \int_{z_0,\overline x^\perp_0,z_K,\overline x^\perp_K} \varphi((z_0,\overline x^\perp_0),(z_K,\overline x^\perp_K)) \mathcal Q((z_0,\overline x^\perp_0),(z_K,\overline x^\perp_K)) \, {\rm d}z_0 \, {\rm d}\overline x^\perp_0 \, {\rm d}z_K \, {\rm d}\overline x^\perp_K \\
&\ \ +\int_{z_0,\overline x^\perp_0,z_K,\overline x^\perp_K} \varphi((z_0,\overline x^\perp_0),(z_0,\overline x^\perp_0)) \left( \nu(z_0,\overline x_0^\perp) - \mathcal Q((z_0,\overline x^\perp_0),(z_K,\overline x^\perp_K))\right) {\rm d}z_0 \, {\rm d}\overline x^\perp_0 \, {\rm d}z_K \, {\rm d}\overline x^\perp_K,
\end{align*}
where $\mathcal Q$ is the density of the couple $(X_n,X_{n+1})$ when the move is accepted. More precisely (notice the identification for the dummy variables $\overline x^\perp_0 \to X_n^\perp$, $\overline x^\perp_K \to \widetilde X_{n+1}^\perp$, $z_0 \to Z_n$, $z_K \to \widetilde Z_{n+1}$),
\begin{align*}
&\mathcal Q((z_0,\overline x^\perp_0),(z_K,\overline x^\perp_K))= \partition^{-1} \exp(-\beta V(z_0,\overline x_0^\perp)) Q(z_0,z_K) \\
&\times \int_{(\overline x^\perp _k)_{1 \le k \le K-1}} 1 \wedge \left(\frac{\exp(-\beta V(z_K,\overline x_K^\perp)) Q(z_K,z_0) \prod_{k=1}^{K} P_{\overline z_k(z_K,z_0)}(\overline x^\perp_{K-k+1}, \overline x^\perp_{K-k})}{\exp(-\beta V(z_0,\overline x_0^\perp)) Q(z_0,z_K) \prod_{k=1}^{K} P_{\overline z_k(z_0,z_K)}(\overline x^\perp_{k-1}, \overline x^\perp_k)} \right) \prod_{k=1}^{K} P_{\overline z_k(z_0,z_K)}(\overline x^\perp_{k-1}, \overline x^\perp_k) \prod_{k=1}^{K-1} {\rm d} \overline x^\perp_k\\
&=\partition^{-1}
\int_{(\overline x^\perp _k)_{1 \le k \le K-1}} \left( \exp(-\beta V(z_0,\overline x_0^\perp)) Q(z_0,z_K) \prod_{k=1}^{K} P_{\overline z_k(z_0,z_K)}(\overline x^\perp_{k-1}, \overline x^\perp_k)\right)  \\
&\qquad \qquad  \qquad \qquad  \wedge \left( \exp(-\beta V(z_K,\overline x_K^\perp)) Q(z_K,z_0) \prod_{k=1}^{K} P_{\overline z_k(z_K,z_0)}(\overline x^\perp_{K-k+1}, \overline x^\perp_{K-k}) \right) \prod_{k=1}^{K-1} {\rm d} \overline x^\perp_k.
\end{align*}
By changing the variables inside the integral as~$(\overline x^\perp _k)_{1 \le k \le K-1} \mapsto (\overline x^\perp _{K-k})_{1 \le k \le K-1}$ and exchanging places between the two term of the minimum, we see that
$$\mathcal Q((z_0,\overline x^\perp_0),(z_K,\overline x^\perp_K))=\mathcal Q((z_K,\overline x^\perp_K),(z_0,\overline x^\perp_0)),$$
which indeed proves the reversibility of the chain: for any bounded measurable function~$\varphi$,
$$\E(\varphi(X_n,X_{n+1}))=\E(\varphi(X_{n+1},X_n)).$$

\subsection{Proof of Proposition \ref{eq:PropChenRoux}}
\label{proof:eq:PropChenRoux}

Assume that $X_n$ is distributed according to $\nu$, and consider the law of $(X_n,X_{n+1})$, where $X_{n+1}$ is obtained from $X_n$ by the algorithm of Section~\ref{sec:ChenRoux}. Recall the notation $X_{n+1}=(Z_{n+1},X^\perp_{n+1})$ and $\widetilde X_{n+1}=(\widetilde Z_{n+1},\widetilde X^\perp_{n+1})$. For any bounded measurable function $\varphi:\R^d \times \R^d \to \R$, one has
\begin{align*}
&  \E(\varphi(X_n,X_{n+1}))\\
&= \E\left[\varphi(X_n,\widetilde X_{n+1}) 1 \wedge \left( \frac{\exp(-\beta V(\widetilde Z_{n+1},\widetilde X_{n+1}^\perp)) Q(\widetilde{Z}_{n+1}, Z_n) \prod_{k=0}^{K-1} P_{\overline z_{k+1-C_n}(\widetilde Z_{n+1}, Z_n )}(\overline X^\perp_{K-k},\overline X^\perp_{K-k-1})}{\exp(-\beta V(Z_n,X_n^\perp)) Q(Z_n,\widetilde{Z}_{n+1}) \prod_{k=1}^{K} P_{\overline z_{k+C_n-1} (Z_n,\widetilde Z_{n+1})}(\overline X^\perp_{k-1},\overline X^\perp_{k})}\right) \right]\\
&  + \E\left[\varphi(X_n, X_{n}) \left(1- 1 \wedge \left( \frac{\exp(-\beta V(\widetilde Z_{n+1},\widetilde X_{n+1}^\perp)) Q(\widetilde{Z}_{n+1}, Z_n) \prod_{k=0}^{K-1} P_{\overline z_{k+1-C_n}(\widetilde Z_{n+1}, Z_n )}(\overline X^\perp_{K-k},\overline X^\perp_{K-k-1})}{\exp(-\beta V(Z_n,X_n^\perp)) Q(Z_n,\widetilde{Z}_{n+1}) \prod_{k=1}^{K} P_{\overline z_{k+C_n-1} (Z_n,\widetilde Z_{n+1})}(\overline X^\perp_{k-1},\overline X^\perp_{k})}\right) \right)\right] \\
&=  \int_{z_0,\overline x^\perp_0,z_K,\overline x^\perp_K} \varphi((z_0,\overline x^\perp_0),(z_K,\overline x^\perp_K)) \mathcal Q((z_0,\overline x^\perp_0),(z_K,\overline x^\perp_K)) \, {\rm d}z_0 \, {\rm d}\overline x^\perp_0 \, {\rm d}z_K \, {\rm d}\overline x^\perp_K \\
& \ \ +\int_{z_0,\overline x^\perp_0,z_K,\overline x^\perp_K} \varphi((z_0,\overline x^\perp_0),(z_0,\overline x^\perp_0)) \left( \nu(z_0,\overline x_0^\perp) - \mathcal Q((z_0,\overline x^\perp_0),(z_K,\overline x^\perp_K))\right)  {\rm d}z_0 \, {\rm d}\overline x^\perp_0 \, {\rm d}z_K \, {\rm d}\overline x^\perp_K,
\end{align*}
where $\mathcal Q$ is the density of the couple $(X_n,X_{n+1})$ when the move is accepted. More precisely (notice the identification for the dummy variables $\overline x^\perp_0 \to X_n^\perp$, $\overline x^\perp_K \to \widetilde X_{n+1}^\perp$, $z_0 \to Z_n$, $z_K \to \widetilde Z_{n+1}$),
\begin{align*}
&\mathcal Q((z_0,\overline x^\perp_0),(z_K,\overline x^\perp_K))= \frac 1 2 \partition^{-1} \exp(-\beta V(z_0,\overline x_0^\perp)) Q(z_0,z_K) \\
&\times \int_{(\overline x^\perp _k)_{1 \le k \le K-1}} 1 \wedge \left(\frac{\exp(-\beta V(z_K,\overline x_K^\perp)) Q(z_K,z_0) \prod_{k=0}^{K-1} P_{\overline z_k(z_K,z_0)}(\overline x^\perp_{K-k}, \overline x^\perp_{K-k-1})}{\exp(-\beta V(z_0,\overline x_0^\perp)) Q(z_0,z_K) \prod_{k=1}^{K} P_{\overline z_k(z_0,z_K)}(\overline x^\perp_{k-1}, \overline x^\perp_k)} \right) \prod_{k=1}^{K} P_{\overline z_k(z_0,z_K)}(\overline x^\perp_{k-1}, \overline x^\perp_k) \prod_{k=1}^{K-1} {\rm d}\overline x^\perp_k\\
& \qquad + \frac 1 2 \partition^{-1} \exp(-\beta V(z_0,\overline x_0^\perp)) Q(z_0,z_K) \\
&\times \int_{(\overline x^\perp _k)_{1 \le k \le K-1}} 1 \wedge \left(\frac{\exp(-\beta V(z_K,\overline x_K^\perp)) Q(z_K,z_0) \prod_{k=0}^{K-1} P_{\overline z_{k+1}(z_K,z_0)}(\overline x^\perp_{K-k}, \overline x^\perp_{K-k-1})}{\exp(-\beta V(z_0,\overline x_0^\perp)) Q(z_0,z_K) \prod_{k=1}^{K} P_{\overline z_{k-1}(z_0,z_K)}(\overline x^\perp_{k-1}, \overline x^\perp_k)} \right) \prod_{k=1}^{K} P_{\overline z_{k-1}(z_0,z_K)}(\overline x^\perp_{k-1}, \overline x^\perp_k) \prod_{k=1}^{K-1} {\rm d}\overline x^\perp_k\\
&=\frac 1 2 \partition^{-1}
\int_{(\overline x^\perp _k)_{1 \le k \le K-1}} \left( \exp(-\beta V(z_0,\overline x_0^\perp)) Q(z_0,z_K) \prod_{k=1}^{K} P_{\overline z_k(z_0,z_K)}(\overline x^\perp_{k-1}, \overline x^\perp_k)\right)  \\
&\qquad \qquad  \qquad \qquad  \wedge \left( \exp(-\beta V(z_K,\overline x_K^\perp)) Q(z_K,z_0) \prod_{k=0}^{K-1} P_{\overline z_k(z_K,z_0)}(\overline x^\perp_{K-k}, \overline x^\perp_{K-k-1}) \right) \prod_{k=1}^{K-1} {\rm d}\overline x^\perp_k\\
&\quad +\frac 1 2 \partition^{-1}
\int_{(\overline x^\perp _k)_{1 \le k \le K-1}} \left( \exp(-\beta V(z_0,\overline x_0^\perp)) Q(z_0,z_K) \prod_{k=1}^{K} P_{\overline z_{k-1}(z_0,z_K)}(\overline x^\perp_{k-1}, \overline x^\perp_k)\right)  \\
&\qquad \qquad  \qquad \qquad  \wedge \left( \exp(-\beta V(z_K,\overline x_K^\perp)) Q(z_K,z_0) \prod_{k=0}^{K-1} P_{\overline z_{k+1}(z_K,z_0)}(\overline x^\perp_{K-k}, \overline x^\perp_{K-k-1}) \right) \prod_{k=1}^{K-1} {\rm d}\overline x^\perp_k.
\end{align*}
As in the previous proofs, one has
\[
\mathcal Q((z_0,\overline x^\perp_0),(z_K,\overline x^\perp_K))=\mathcal Q((z_K,\overline x^\perp_K),(z_0,\overline x^\perp_0)),
\]
To see this, consider $Q((z_K,\overline x^\perp_K),(z_0,\overline x^\perp_0))$ and change the integration variables as $(\overline x^\perp _k)_{1 \le k \le K-1} \mapsto (\overline x^\perp _{K-k})_{1 \le k \le K-1}$. Swapping the two terms in the sum and swapping the terms within each minimum, after an additional index shift in the products, one ends up with exactly the same expression as the one for $Q((z_0,\overline x^\perp_0),(z_K,\overline x^\perp_K))$ given above.
This indeed proves the reversibility of the chain: for any bounded measurable function~$\varphi$,
$$\E(\varphi(X_n,X_{n+1}))=\E(\varphi(X_{n+1},X_n)).$$

\subsection{Proof of Proposition \ref{eq:propchenrouxwork}} 
\label{proof:eq:propchenrouxwork}

We first consider the case $C_n=1$. One then has (forgetting the term $\frac{Q(\widetilde{Z}_{n+1}, Z_n)}{Q( Z_n,\widetilde{Z}_{n+1})}$ which appears in both the left and right-hand sides of~\eqref{eq:work}):
\begin{align*}
&\frac{\exp(-\beta V(\widetilde Z_{n+1},\widetilde X_{n+1}^\perp))  \prod_{k=0}^{K-1} P_{\overline z_k(\widetilde Z_{n+1}, Z_n )}(\overline X^\perp_{K-k},\overline X^\perp_{K-k-1})}{\exp(-\beta V(Z_n,X_n^\perp))  \prod_{k=1}^{K} P_{\overline z_k(Z_n,\widetilde Z_{n+1})}(\overline X^\perp_{k-1},\overline X^\perp_{k})}\\
&=   \frac{\exp(-\beta V(\widetilde Z_{n+1},\widetilde X_{n+1}^\perp))  \prod_{k=0}^{K-1} P_{\overline z_{K-k}( Z_n,\widetilde Z_{n+1} )}(\overline X^\perp_{K-k},\overline X^\perp_{K-k-1})}{\exp(-\beta V(Z_n,X_n^\perp))  \prod_{k=1}^{K} P_{\overline z_k(Z_n,\widetilde Z_{n+1})}(\overline X^\perp_{k-1},\overline X^\perp_{k})}\\
&=   \frac{\exp(-\beta V(\widetilde Z_{n+1},\widetilde X_{n+1}^\perp))  \prod_{k=1}^{K} P_{\overline z_{k}( Z_n,\widetilde Z_{n+1} )}(\overline X^\perp_{k},\overline X^\perp_{k-1})}{\exp(-\beta V(Z_n,X_n^\perp))  \prod_{k=1}^{K} P_{\overline z_k(Z_n,\widetilde Z_{n+1})}(\overline X^\perp_{k-1},\overline X^\perp_{k})}=   \frac{\exp(-\beta V(\widetilde Z_{n+1},\widetilde X_{n+1}^\perp))}{\exp(-\beta V(Z_n,X_n^\perp)) }  \prod_{k=1}^{K} \frac{P_{\overline z_{k}( Z_n,\widetilde Z_{n+1} )}(\overline X^\perp_{k},\overline X^\perp_{k-1})}{ P_{\overline z_k(Z_n,\widetilde Z_{n+1})}(\overline X^\perp_{k-1},\overline X^\perp_{k})}\\
&=   \frac{\exp(-\beta V(\overline z_K(Z_n,\widetilde Z_{n+1} ) ,\overline X_{K}^\perp))}{\exp(-\beta V(\overline z_0(Z_n,\widetilde Z_{n+1} ),\overline X_0^\perp)) }  \prod_{k=1}^{K} \frac{\exp(-\beta V(\overline z_{k}( Z_n,\widetilde Z_{n+1}),\overline X^\perp_{k-1}))}{ \exp(-\beta V(\overline z_{k}( Z_n,\widetilde Z_{n+1}),\overline X^\perp_{k}))} = \frac{\prod_{k=0}^{K-1}\exp(-\beta V(\overline z_{k+1}( Z_n,\widetilde Z_{n+1}),\overline X^\perp_{k}))}{ \prod_{k=0}^{K-1} \exp(-\beta V(\overline z_{k}( Z_n,\widetilde Z_{n+1} ),\overline X^\perp_{k}))}\\
&= \exp(-\beta {\mathcal W}^1_{n+1}),
\end{align*} 
where we used~\eqref{eq:rev_chemin_prime} in the first equality, a change of indices ($K-k \to k$) in the numerator in the second equality, and~\eqref{eq:rev_second} in the fourth one. 

Likewise, if $C_n=0$, one has
\begin{align*}
&\frac{\exp(-\beta V(\widetilde Z_{n+1},\widetilde X_{n+1}^\perp))  \prod_{k=0}^{K-1} P_{\overline z_{k+1}(\widetilde Z_{n+1}, Z_n )}(\overline X^\perp_{K-k},\overline X^\perp_{K-k-1})}{\exp(-\beta V(Z_n,X_n^\perp))  \prod_{k=1}^{K} P_{\overline z_{k-1}(Z_n,\widetilde Z_{n+1})}(\overline X^\perp_{k-1},\overline X^\perp_{k})}\\
&=   \frac{\exp(-\beta V(\widetilde Z_{n+1},\widetilde X_{n+1}^\perp))  \prod_{k=0}^{K-1} P_{\overline z_{K-k-1}( Z_n,\widetilde Z_{n+1} )}(\overline X^\perp_{K-k},\overline X^\perp_{K-k-1})}{\exp(-\beta V(Z_n,X_n^\perp))  \prod_{k=1}^{K} P_{\overline z_{k-1}(Z_n,\widetilde Z_{n+1})}(\overline X^\perp_{k-1},\overline X^\perp_{k})}\\
&=   \frac{\exp(-\beta V(\widetilde Z_{n+1},\widetilde X_{n+1}^\perp))  \prod_{k=1}^{K} P_{\overline z_{k-1}( Z_n,\widetilde Z_{n+1} )}(\overline X^\perp_{k},\overline X^\perp_{k-1})}{\exp(-\beta V(Z_n,X_n^\perp))  \prod_{k=1}^{K} P_{\overline z_{k-1}(Z_n,\widetilde Z_{n+1})}(\overline X^\perp_{k-1},\overline X^\perp_{k})}=  \frac{\exp(-\beta V(\widetilde Z_{n+1},\widetilde X_{n+1}^\perp))}{\exp(-\beta V(Z_n,X_n^\perp)) }  \prod_{k=1}^{K} \frac{P_{\overline z_{k-1}( Z_n,\widetilde Z_{n+1} )}(\overline X^\perp_{k},\overline X^\perp_{k-1})}{ P_{\overline z_{k-1}(Z_n,\widetilde Z_{n+1})}(\overline X^\perp_{k-1},\overline X^\perp_{k})}\\
&=   \frac{\exp(-\beta V(\overline z_K(Z_n,\widetilde Z_{n+1} ) ,\overline X_{K}^\perp))}{\exp(-\beta V(\overline z_0(Z_n,\widetilde Z_{n+1} ),\overline X_0^\perp)) }  \prod_{k=1}^{K} \frac{\exp(-\beta V(\overline z_{k-1}( Z_n,\widetilde Z_{n+1} ),\overline X^\perp_{k-1}))}{ \exp(-\beta V(\overline z_{k-1}( Z_n,\widetilde Z_{n+1} ),\overline X^\perp_{k}))}
=   \frac{\prod_{k=1}^{K}\exp(-\beta V(\overline z_{k}( Z_n,\widetilde Z_{n+1} ),\overline X^\perp_{k}))}{ \prod_{k=1}^{K} \exp(-\beta V(\overline z_{k-1}( Z_n,\widetilde Z_{n+1}),\overline X^\perp_{k}))}\\
&= \exp(-\beta {\mathcal W}^0_{n+1}).
\end{align*} 

To conclude the proof, let us prove~\eqref{eq:Jarz_consistency} in this setting (recall that we suppose $X^\perp_n$ at equilibrium). If $C_n=1$, for any bounded measurable function~$\varphi:\R^{d-1} \to \R$ (we omit here to indicate explicitly the dependence of $\overline z_k$ on ($Z_n,\widetilde Z_{n+1})$),
\begin{align*}
& \E\left(\varphi(\widetilde X^\perp_{n+1}) \exp(-\beta \mathcal W^1_{n+1}) \right) \\
& = \int_{(\overline x_k^\perp)_{0 \le k \le K}} \varphi(\overline x_K^\perp) \prod_{k=0}^{K-1} \exp(-\beta [V(\overline z_{k+1},\overline x_k^\perp)-V(\overline z_{k},\overline x_k^\perp)]) \frac{\exp(-\beta V(\overline z_0,\overline x_0^\perp))}{\partition \exp(-\beta F(\overline z_0))} \prod_{k=1}^{K}P_{\overline z_k}(\overline x^\perp_{k-1},{\rm d} \overline x^\perp_k) \, {\rm d} \overline x^\perp_0\\
&=\int_{(\overline x_k^\perp)_{1 \le k \le K}} \varphi(\overline x_K^\perp) \prod_{k=1}^{K-1} \exp(-\beta (V(\overline  z_{k+1},\overline x_k^\perp)-V(\overline z_{k},\overline x_k^\perp))) \int_{\overline x_0^\perp} \frac{\exp(-\beta V(\overline z_1,\overline x_0^\perp))}{\partition \exp(-\beta F(\overline z_0))}  P_{\overline z_1}(\overline x^\perp_{0},{\rm d}\overline x^\perp_1) \, {\rm d}\overline x^\perp_0\prod_{k=2}^{K}P_{\overline z_k}(\overline x^\perp_{k-1},{\rm d}\overline x^\perp_k) \\
&=\int_{(\overline x_k^\perp)_{1 \le k \le K}} \varphi(\overline x_K^\perp) \prod_{k=1}^{K-1} \exp(-\beta (V(\overline z_{k+1},\overline x_k^\perp)-V(\overline z_{k},\overline x_k^\perp))) \frac{\exp(-\beta V(\overline z_1,\overline x_1^\perp))}{\partition \exp(-\beta F(\overline z_0))} \, {\rm d} \overline x^\perp_1 \prod_{k=2}^{K}P_{\overline z_k}(\overline x^\perp_{k-1},{\rm d} \overline x^\perp_k), 
\end{align*}
where we used the invariance of $\exp(-\beta V(\overline z_1,x^\perp)) {\rm d}x^\perp / \exp(-\beta F(\overline z_1))$ by the kernel $P_{\overline z_1}(x^\perp,{\rm d}\overline x^\perp)$. By iterating the argument, one obtains:
\begin{align*}
  & \E\left(\varphi(\widetilde X^\perp_{n+1}) \exp(-\beta \mathcal W^1_{n+1}) \right) \\
  &=\int_{(\overline x_k^\perp)_{K-1 \le k \le K}} \varphi(\overline x_K^\perp)  \exp(-\beta [V(\overline z_{K},\overline x_{K-1}^\perp)-V(\overline z_{K-1},\overline x_{K-1}^\perp)]) \frac{\exp(-\beta V(\overline z_{K-1},\overline x_{K-1}^\perp))}{\partition \exp(-\beta F(\overline z_0))} \,{\rm d} \overline x^\perp_{K-1} P_{\overline z_K}(\overline x^\perp_{K-1},{\rm d} \overline x^\perp_K)\\
  &=\int_{(\overline x_k^\perp)_{K-1 \le k \le K}} \varphi(\overline x_K^\perp) \frac{\exp(-\beta V(\overline z_{K},\overline x_{K-1}^\perp))}{\partition \exp(-\beta F(\overline z_0))} \,{\rm d} \overline x^\perp_{K-1} P_{\overline z_K}(\overline x^\perp_{K-1},{\rm d} \overline x^\perp_K) \\
  &=\int_{\overline x_K^\perp} \varphi(\overline x_K^\perp) \frac{\exp(-\beta V(\overline z_{K},\overline x_{K}^\perp))}{\partition \exp(-\beta F(\overline z_0))} \, {\rm d}\overline x^\perp_{K} \\
  &= \exp(-\beta (F(\overline z_K)-F(\overline z_0))) \int_{\overline x_K^\perp} \varphi(\overline x_K^\perp) \frac{\exp(-\beta V(\overline z_{K},\overline x_{K}^\perp))}{\partition \exp(-\beta F(\overline z_K))} \, {\rm d} \overline x^\perp_{K},
\end{align*}
which is exactly~\eqref{eq:Jarz_consistency}. For~$C_n=0$, 
\begin{align*}
& \E\left(\varphi(\widetilde X^\perp_{n+1}) \exp(-\beta \mathcal W^1_{n+1}) \right) \\
&=\int_{(\overline x_k^\perp)_{0 \le k \le K}} \varphi(\overline x_K^\perp) \prod_{k=0}^{K-1} \exp(-\beta [V(\overline z_{k+1},\overline x_{k+1}^\perp)-V(\overline z_{k},\overline x_{k+1}^\perp)]) \frac{\exp(-\beta V(\overline z_0,\overline x_0^\perp))}{\partition \exp(-\beta F(\overline z_0))} \prod_{k=1}^{K}P_{\overline z_{k-1}}(\overline x^\perp_{k-1},{\rm d} \overline x^\perp_k) \, {\rm d} \overline x^\perp_0\\
&=\int_{(\overline x_k^\perp)_{1 \le k \le K}} \varphi(\overline x_K^\perp) \prod_{k=0}^{K-1} \exp(-\beta [V(\overline  z_{k+1},\overline x_{k+1}^\perp)-V(\overline z_{k},\overline x_{k+1}^\perp)]) \int_{\overline x_0^\perp} \frac{\exp(-\beta V(\overline z_0,\overline x_0^\perp))}{\partition \exp(-\beta F(\overline z_0))}  P_{\overline z_0}(\overline x^\perp_{0},{\rm d}\overline x^\perp_1) \, {\rm d}\overline x^\perp_0\prod_{k=2}^{K}P_{\overline z_{k-1}}(\overline x^\perp_{k-1},{\rm d} \overline x^\perp_k) \\
&=\int_{(\overline x_k^\perp)_{1 \le k \le K}} \varphi(\overline x_K^\perp) \prod_{k=0}^{K-1} \exp(-\beta [V(\overline z_{k+1},\overline x_{k+1}^\perp)-V(\overline z_{k},\overline x_{k+1}^\perp)]) \frac{\exp(-\beta V(\overline z_0,\overline x_1^\perp))}{\partition \exp(-\beta F(\overline z_0))} \, {\rm d}\overline x^\perp_1 \prod_{k=2}^{K}P_{\overline z_{k-1}}(\overline x^\perp_{k-1},{\rm d}\overline x^\perp_k) \\
&=\int_{(\overline x_k^\perp)_{1 \le k \le K}} \varphi(\overline x_K^\perp) \prod_{k=1}^{K-1} \exp(-\beta [V(\overline z_{k+1},\overline x_{k+1}^\perp)-V(\overline z_{k},\overline x_{k+1}^\perp)]) \frac{\exp(-\beta V(\overline z_1,\overline x_1^\perp))}{\partition \exp(-\beta F(\overline z_0))} \, {\rm d}\overline x^\perp_1 \prod_{k=2}^{K}P_{\overline z_{k-1}}(\overline x^\perp_{k-1},{\rm d} \overline x^\perp_k), 
\end{align*}
where we used the invariance of $\exp(-\beta V(\overline z_1,x^\perp)) {\rm d}x^\perp / \exp(-\beta F(\overline z_1))$ by the kernel $P_{\overline z_1}(x^\perp,{\rm d}\overline x^\perp)$. By iterating the argument, one obtains:
\begin{align*}
& \E\left(\varphi(\widetilde X^\perp_{n+1}) \exp(-\beta \mathcal W^1_{n+1}) \right) \\
&=\int_{(\overline x_k^\perp)_{K-1 \le k \le K}} \varphi(\overline x_K^\perp)  \exp(-\beta [V(\overline z_{K},\overline x_{K}^\perp)-V(\overline z_{K-1},\overline x_{K}^\perp)]) \frac{\exp(-\beta V(\overline z_{K-1},\overline x_{K-1}^\perp))}{\partition \exp(-\beta F(\overline z_0))} \, {\rm d}\overline x^\perp_{K-1} P_{\overline z_{K-1}}(\overline x^\perp_{K-1},{\rm d}\overline x^\perp_K) \\
&=\int_{\overline x_K^\perp} \varphi(\overline x_K^\perp)  \exp(-\beta [V(\overline z_{K},\overline x_{K}^\perp)-V(\overline z_{K-1},\overline x_{K}^\perp)]) \frac{\exp(-\beta V(\overline z_{K-1},\overline x_{K}^\perp))}{\partition \exp(-\beta F(\overline z_0))} \, {\rm d}\overline x^\perp_{K}  \\
&=\int_{\overline x_K^\perp} \varphi(\overline x_K^\perp) \frac{\exp(-\beta V(\overline z_{K},\overline x_{K}^\perp))}{\partition \exp(-\beta F(\overline z_0))} \, {\rm d}\overline x^\perp_{K} \\
&= \exp(-\beta [F(\overline z_K)-F(\overline z_0)]) \int_{\overline x_K^\perp} \varphi(\overline x_K^\perp) \frac{\exp(-\beta V(\overline z_{K},\overline x_{K}^\perp))}{\partition \exp(-\beta F(\overline z_K))} \, {\rm d}\overline x^\perp_{K},
\end{align*}
which is exactly~\eqref{eq:Jarz_consistency}.

\subsection{Proof of Proposition \ref{eq:propstrangrev}}
\label{proof:eq:propstrangrev}

Assume that~$X_n$ is distributed according to~$\nu$, and consider the law of~$(X_n,X_{n+1})$, where~$X_{n+1}$ is obtained from~$X_n$ by the algorithm of \cref{sec:symm}. Recall the notation~$X_{n+1}=(Z_{n+1},X^\perp_{n+1})$ and $\widetilde X_{n+1}=(\widetilde Z_{n+1},\widetilde X^\perp_{n+1})$. For any bounded measurable function $\varphi:\R^d \times \R^d \to \R$, one has
\begin{align*}
&  \E(\varphi(X_n,X_{n+1}))\\
&= \E\left[\varphi(X_n,\widetilde X_{n+1}) 1 \wedge \left( \frac{\exp(-\beta V(\widetilde Z_{n+1},\widetilde X_{n+1}^\perp)) Q(\widetilde{Z}_{n+1}, Z_n) \prod_{k=0}^K P_{\overline z_k(\widetilde Z_{n+1},Z_n)}(\overline X^\perp_{K+1-k},\overline X^\perp_{K-k})}{\exp(-\beta V(Z_n,X_n^\perp)) Q(Z_n,\widetilde{Z}_{n+1}) \prod_{k=0}^K P_{\overline z_k(Z_n,\widetilde Z_{n+1})}(\overline X^\perp_{k},\overline X^\perp_{k+1})}\right) \right]\\
&  + \E\left[\varphi(X_n, X_{n}) \left(1- 1 \wedge \left( \frac{\exp(-\beta V(\widetilde Z_{n+1},\widetilde X_{n+1}^\perp)) Q(\widetilde{Z}_{n+1}, Z_n) \prod_{k=0}^K P_{\overline z_k(\widetilde Z_{n+1},Z_n)}(\overline X^\perp_{K+1-k},\overline X^\perp_{K-k})}{\exp(-\beta V(Z_n,X_n^\perp)) Q(Z_n,\widetilde{Z}_{n+1}) \prod_{k=0}^K P_{\overline z_k(Z_n,\widetilde Z_{n+1})}(\overline X^\perp_{k},\overline X^\perp_{k+1})}\right) \right)\right] \\
&=  \int_{z_0,\overline x^\perp_0,z_K,\overline x^\perp_{K+1}} \varphi((z_0,\overline x^\perp_0),(z_K,\overline x^\perp_{K+1})) \mathcal Q((z_0,\overline x^\perp_0),(z_K,\overline x^\perp_{K+1})) \, {\rm d}z_0 \, {\rm d}\overline x^\perp_0 \, {\rm d} z_K \, {\rm d}\overline x^\perp_{K+1} \\
&\ \ +\int_{z_0,\overline x^\perp_0,z_K,\overline x^\perp_{K+1}} \varphi((z_0,\overline x^\perp_0),(z_0,\overline x^\perp_0)) \left( \nu(z_0,\overline x_0^\perp) - \mathcal Q((z_0,\overline x^\perp_0),(z_K,\overline x^\perp_{K+1}))\right) \, {\rm d}z_0 \, {\rm d}\overline x^\perp_0 \, {\rm d} z_K \, {\rm d}\overline x^\perp_{K+1} ,
\end{align*}
where $\mathcal Q$ is the density of the couple $(X_n,X_{n+1})$ when the move is accepted. More precisely (notice the identification for the dummy variables $\overline x^\perp_0 \to X_n^\perp$, $\overline x^\perp_{K+1} \to \widetilde X_{n+1}^\perp$, $z_0 \to Z_n$, $z_K \to \widetilde Z_{n+1}$),
\begin{align*}
&\mathcal Q((z_0,\overline x^\perp_0),(z_K,\overline x^\perp_{K+1}))= \partition^{-1} \exp(-\beta V(z_0,\overline x_0^\perp)) Q(z_0,z_K) \\
&\times \int_{(\overline x^\perp _k)_{1 \le k \le K}} \left[1 \wedge \left(\frac{\displaystyle \exp(-\beta V(z_K,\overline x^\perp_{K+1})) Q(z_K,z_0) \prod_{k=0}^{K} P_{\overline z_k(z_K,z_0)}(\overline x^\perp_{K+1-k}, \overline x^\perp_{K-k})}{\displaystyle \exp(-\beta V(z_0,\overline x_0^\perp)) Q(z_0,z_K) \prod_{k=0}^{K} P_{\overline z_k(z_0,z_K)}(\overline x^\perp_{k}, \overline x^\perp_{k+1})} \right)\right] \prod_{k=0}^{K} P_{\overline z_k(z_0,z_K)}(\overline x^\perp_{k}, \overline x^\perp_{k+1}) \, \prod_{k=1}^{K} {\rm d}\overline x^\perp_k\\
&=\partition^{-1}
\int_{(\overline x^\perp _k)_{1 \le k \le K}} \left( \exp(-\beta V(z_0,\overline x_0^\perp)) Q(z_0,z_K) \prod_{k=0}^{K} P_{\overline z_k(z_0,z_K)}(\overline x^\perp_{k}, \overline x^\perp_{k+1})\right)  \\
&\qquad \qquad  \qquad \qquad  \wedge \left( \exp(-\beta V(z_K,\overline x^\perp_{K+1})) Q(z_K,z_0) \prod_{k=0}^{K} P_{\overline z_k(z_K,z_0)}(\overline x^\perp_{K+1-k}, \overline x^\perp_{K-k}) \right) \, \prod_{k=1}^{K} {\rm d}\overline x^\perp_k.
\end{align*}
By renaming the variables inside the integral as $(\overline x^\perp _k)_{1 \le k \le K} \mapsto (\overline x^\perp _{K+1-k})_{1 \le k \le K}$, we see that
\[
\mathcal Q((z_0,\overline x^\perp_0),(z_K,\overline x^\perp_{K+1}))=\mathcal Q((z_K,\overline x^\perp_{K+1}),(z_0,\overline x^\perp_0)),
\]
which indeed proves the reversibility of the chain: for any bounded measurable function~$\varphi$,
\[
\E(\varphi(X_n,X_{n+1}))=\E(\varphi(X_{n+1},X_n)).
\]

\subsection{Proof of Proposition \ref{eq:propStrangWorks}}
\label{proof:eq:propStrangWorks}

One has (forgetting the term $\frac{Q(\widetilde{Z}_{n+1}, Z_n)}{Q( Z_n,\widetilde{Z}_{n+1})}$ which appears in both the left and right-hand sides of~\eqref{eq:work2}):
\begin{align*}
&\frac{\exp(-\beta V(\widetilde Z_{n+1},\widetilde X_{n+1}^\perp))  \prod_{k=0}^{K} P_{\overline z_k(\widetilde Z_{n+1}, Z_n )}(\overline X^\perp_{K+1-k},\overline X^\perp_{K-k})}{\exp(-\beta V(Z_n,X_n^\perp))  \prod_{k=0}^{K} P_{\overline z_k(Z_n,\widetilde Z_{n+1})}(\overline X^\perp_{k},\overline X^\perp_{k+1})}\\
&=   \frac{\exp(-\beta V(\widetilde Z_{n+1},\widetilde X_{n+1}^\perp))  \prod_{k=0}^{K} P_{\overline z_{K-k}( Z_n,\widetilde Z_{n+1} )}(\overline X^\perp_{K+1-k},\overline X^\perp_{K-k})}{\exp(-\beta V(Z_n,X_n^\perp))  \prod_{k=0}^{K} P_{\overline z_k(Z_n,\widetilde Z_{n+1})}(\overline X^\perp_{k},\overline X^\perp_{k+1})}\\
&=   \frac{\exp(-\beta V(\widetilde Z_{n+1},\widetilde X_{n+1}^\perp))  \prod_{k=0}^{K} P_{\overline z_{k}( Z_n,\widetilde Z_{n+1} )}(\overline X^\perp_{k+1},\overline X^\perp_{k})}{\exp(-\beta V(Z_n,X_n^\perp))  \prod_{k=0}^{K} P_{\overline z_k(Z_n,\widetilde Z_{n+1})}(\overline X^\perp_{k},\overline X^\perp_{k+1})}
=   \frac{\exp(-\beta V(\widetilde Z_{n+1},\widetilde X_{n+1}^\perp))}{\exp(-\beta V(Z_n,X_n^\perp)) }  \prod_{k=0}^{K} \frac{P_{\overline z_{k}( Z_n,\widetilde Z_{n+1} )}(\overline X^\perp_{k+1},\overline X^\perp_{k})}{ P_{\overline z_k(Z_n,\widetilde Z_{n+1})}(\overline X^\perp_{k},\overline X^\perp_{k+1})}\\
&=   \frac{\exp(-\beta V(\overline z_K(Z_n,\widetilde Z_{n+1} ) ,\overline X_{K+1}^\perp))}{\exp(-\beta V(\overline z_0(Z_n,\widetilde Z_{n+1} ),\overline X_0^\perp)) }  \prod_{k=0}^{K} \frac{\exp(-\beta V(\overline z_{k}( Z_n,\widetilde Z_{n+1}),\overline X^\perp_{k}))}{ \exp(-\beta V(\overline z_{k}( Z_n,\widetilde Z_{n+1}),\overline X^\perp_{k+1}))}\\
&=     \frac{\prod_{k=1}^{K}\exp(-\beta V(\overline z_{k}( Z_n,\widetilde Z_{n+1} ),\overline X^\perp_{k}))}{ \prod_{k=0}^{K-1} \exp(-\beta V(\overline z_{k}( Z_n,\widetilde Z_{n+1} ),\overline X^\perp_{k+1}))}
=     \frac{\prod_{k=1}^{K}\exp(-\beta V(\overline z_{k}( Z_n,\widetilde Z_{n+1}),\overline X^\perp_{k}))}{ \prod_{k=1}^{K} \exp(-\beta V(\overline z_{k-1}( Z_n,\widetilde Z_{n+1}),\overline X^\perp_{k}))}\\
&= \exp(-\beta {\mathcal W}^2_{n+1}),
\end{align*} 
where we used~\eqref{eq:rev_chemin_prime} in the first equality, a change of indices ($K-k \to k$) in the numerator in the second equality, and~\eqref{eq:rev_second} in the fourth one. 

To conclude the proof, let us prove~\eqref{eq:Jarz_consistency} in this setting (recall that we suppose $X^\perp_n$ at equilibrium). For any bounded measurable function~$\varphi$ (we omit here to indicate explicitly the dependence of $\overline z_k$ on ($Z_n,\widetilde Z_{n+1})$):
\begin{align*}
& \E\left(\varphi(\widetilde X^\perp_{n+1}) \exp(-\beta \mathcal W^2_{n+1}) \right) \\
&=\int_{(\overline x_k^\perp)_{0 \le k \le K+1}} \varphi(\overline x_{K+1}^\perp) \prod_{k=1}^{K} \exp(-\beta [V(\overline z_{k},\overline x_k^\perp)-V(\overline z_{k-1},\overline x_k^\perp)]) \frac{\exp(-\beta V(\overline z_0,\overline x_0^\perp))}{\partition \exp(-\beta F(\overline z_0))} \prod_{k=0}^{K}P_{\overline z_k}(\overline x^\perp_{k},{\rm d} \overline x^\perp_{k+1}) \, {\rm d}\overline x^\perp_0\\
&=\int_{(\overline x_k^\perp)_{1 \le k \le K+1}} \varphi(\overline x_{K+1}^\perp) \prod_{k=1}^{K} \exp(-\beta [V(\overline  z_{k},\overline x_k^\perp)-V(\overline z_{k-1},\overline x_k^\perp)]) \int_{\overline x_0^\perp} \frac{\exp(-\beta V(\overline z_0,\overline x_0^\perp))}{\partition \exp(-\beta F(\overline z_0))}  P_{\overline z_0}(\overline x^\perp_{0},{\rm d}\overline x^\perp_1) \, {\rm d}\overline x^\perp_0\prod_{k=1}^{K}P_{\overline z_k}(\overline x^\perp_{k},{\rm d}\overline x^\perp_{k+1}) \\
&=\int_{(\overline x_k^\perp)_{1 \le k \le K+1}} \varphi(\overline x_{K+1}^\perp) \prod_{k=1}^{K} \exp(-\beta [V(\overline z_{k},\overline x_k^\perp)-V(\overline z_{k-1},\overline x_k^\perp)]) \frac{\exp(-\beta V(\overline z_0,\overline x_1^\perp))}{\partition \exp(-\beta F(\overline z_0))} \, {\rm d}\overline x^\perp_1 \prod_{k=1}^{K}P_{\overline z_k}(\overline x^\perp_{k},{\rm d}\overline x^\perp_{k+1}) \\
&=\int_{(\overline x_k^\perp)_{1 \le k \le K+1}} \varphi(\overline x_{K+1}^\perp) \prod_{k=2}^{K} \exp(-\beta [V(\overline z_{k},\overline x_k^\perp)-V(\overline z_{k-1},\overline x_k^\perp)]) \frac{\exp(-\beta V(\overline z_1,\overline x_1^\perp))}{\partition \exp(-\beta F(\overline z_0))} \, {\rm d}\overline x^\perp_1 \prod_{k=1}^{K}P_{\overline z_k}(\overline x^\perp_{k},{\rm d}\overline x^\perp_{k+1}),
\end{align*}
where we used the invariance of $\exp(-\beta V(\overline z_0,x^\perp)) \, {\rm d}x^\perp / \exp(-\beta F(\overline z_0))$ by the kernel $P_{\overline z_0}(x^\perp,{\rm d}\overline x^\perp)$. By iterating the argument, one obtains:
\begin{align*}
& \E\left(\varphi(\widetilde X^\perp_{n+1}) \exp(-\beta \mathcal W^2_{n+1}) \right) \\
  &=\int_{(\overline x_{K-1}^\perp,\overline x_K^\perp,\overline x_{K+1}^\perp)}  \varphi(\overline x_{K+1}^\perp)  \exp(-\beta [V(\overline z_{K},\overline x_{K}^\perp)-V(\overline z_{K-1},\overline x_{K}^\perp)]) \frac{\exp(-\beta V(\overline z_{K-1},\overline x_{K-1}^\perp))}{\partition \exp(-\beta F(\overline z_0))} \, {\rm d}\overline x^\perp_{K-1} \\
  & \qquad\qquad\qquad\qquad\qquad\qquad\qquad\qquad\qquad\qquad\qquad\qquad \times P_{\overline z_{K-1}}(\overline x^\perp_{K-1},{\rm d}\overline x^\perp_K) P_{\overline z_K}(\overline x^\perp_{K},{\rm d}\overline x^\perp_{K+1}) \\
  &=\int_{(\overline x_K^\perp,\overline x_{K+1}^\perp)} \varphi(\overline x_{K+1}^\perp)  \exp(-\beta [V(\overline z_{K},\overline x_{K}^\perp)-V(\overline z_{K-1},\overline x_{K}^\perp)]) \int_{\overline x_{K-1}^\perp}\frac{\exp(-\beta V(\overline z_{K-1},\overline x_{K-1}^\perp))}{\partition \exp(-\beta F(\overline z_0))} \, {\rm d}\overline x^\perp_{K-1} \\
  & \qquad\qquad\qquad\qquad\qquad\qquad\qquad\qquad\qquad\qquad\qquad\qquad \times P_{\overline z_{K-1}}(\overline x^\perp_{K-1},{\rm d}\overline x^\perp_K) P_{\overline z_K}(\overline x^\perp_{K},{\rm d}\overline x^\perp_{K+1}) \\
&=\int_{(\overline x_K^\perp,\overline x_{K+1}^\perp)} \varphi(\overline x_{K+1}^\perp)  \exp(-\beta [V(\overline z_{K},\overline x_{K}^\perp)-V(\overline z_{K-1},\overline x_{K}^\perp)]) \frac{\exp(-\beta V(\overline z_{K-1},\overline x_{K}^\perp))}{\partition \exp(-\beta F(\overline z_0))} \, {\rm d}\overline x^\perp_{K}  P_{\overline z_K}(\overline x^\perp_{K},{\rm d}\overline x^\perp_{K+1}) \\
&=\int_{(\overline x_K^\perp,\overline x_{K+1}^\perp)} \varphi(\overline x_{K+1}^\perp) \frac{\exp(-\beta V(\overline z_{K},\overline x_{K}^\perp))}{\partition \exp(-\beta F(\overline z_0))} \, {\rm d}\overline x^\perp_{K}  P_{\overline z_K}(\overline x^\perp_{K},{\rm d}\overline x^\perp_{K+1}) \\
&=\int_{\overline x_{K+1}^\perp} \varphi(\overline x_{K+1}^\perp)  
 \frac{\exp(-\beta V(\overline z_{K},\overline x_{K+1}^\perp))}{\partition \exp(-\beta F(\overline z_0))} \, {\rm d}\overline x^\perp_{K+1} \\
& = \exp(-\beta [F(\overline z_K)-F(\overline z_0)]) \int_{\overline x_{K+1}^\perp} \varphi(\overline x_{K+1}^\perp)  
 \frac{\exp(-\beta V(\overline z_{K},\overline x_{K+1}^\perp))}{\partition \exp(-\beta F(\overline z_K))} \, {\rm d}\overline x^\perp_{K},
\end{align*}
which is exactly~\eqref{eq:Jarz_consistency}.

\clearpage

\section{Choice of collective variable}
\label{sec:app_cv_choice_failure}

 The algorithms presented in this work rely on a suitable choice of the collective variable(s). Most of the examples we consider in the main text are such that the conditional distributions $\nu_\perp({\rm d}x^\perp|z)$ are unimodal.   
 In \cref{fig:app_cv_choice_failure}, we contrast this with an extreme negative example where the conditional measure is bimodal for some values of $z$, with a mode switching as the $z$ values evolve. Then, the equilibration steps fail to perform a switch to the other mode that appears when following the CV schedule (akin to a first order phase transition in the language of statistical mechanics). Due to lack of equilibration, the proposed moves following from these constructed paths will have vanishing acceptance probabilities and the algorithm will fail.
Note that metastability in the conditional measures does not necessarily prevent transitions as we demonstrate on the $\phi^4$ model of \cref{sec:phi4}.

\begin{figure}[h!]
    \centering
    \includegraphics{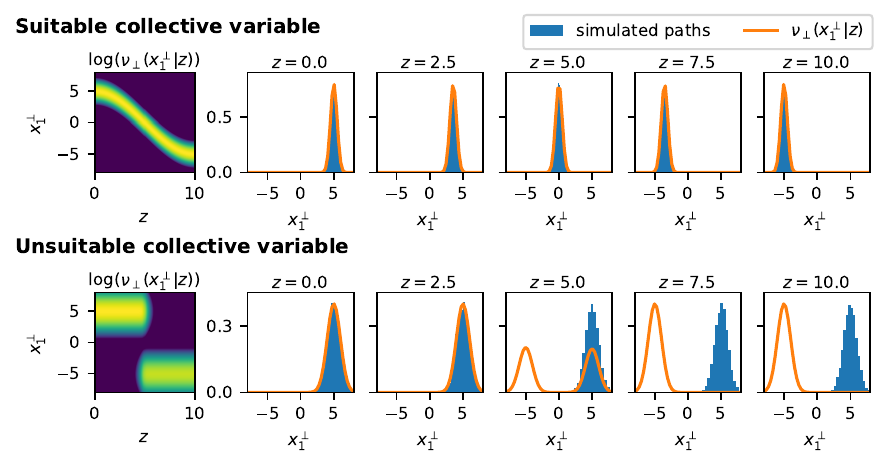}
    \caption{Positive and negative examples of constructed transitions paths for a suitable and an unsuitable collective variable. Both examples are for total dimension $d=20$. The example in the upper panel is the one described in \cref{subsec:fixedmodeswitch}, the lower panel example is a Gaussian mixture model with $\nu=\mathcal{N}(\mu_1,I)/2+\mathcal{N}(\mu_2,I)/2$, $\mu_1=(0,5,\dots,5)^T$ and $\mu_2=(10,-5,\dots,-5)^T$. With the collective variable $z=x_1$ used throughout the numerical section, paths from $\overline{z}_0=0$ to $\overline{z}_K=10$ were constructed using a Mala resampling kernel with $\Delta t = 0.4$ along a CV schedule with $K=3200$ steps. Results are shown for 10,000 chains, where the coordinate $x^\perp$ was randomly sampled from the conditional distribution at the initial point $\nu_\perp(x^\perp|\overline{z}_0)$. For the SymMala algorithm, and assuming a symmetric kernel $Q(z,
    \widetilde{z})$ such that is does not enter into the acceptance criterion, the constructed paths in the upper panel have an average acceptance rate of $\sim30\%$, whereas the paths in the lower panel are all rejected without exception.}
    \label{fig:app_cv_choice_failure}
\end{figure}

\section{Details of the Wavelet decomposition}
\label{sec:wavelet}

We use a Haar wavelet decomposition and just a give a very practical definition here. For a more extensive introduction, we refer the reader to~\cite{chui1992}. For a one dimensional input signal $\{\phi_i\}_{i=1}^N$ of size $N=2^n$, set $\varphi^{(0)}_i = \phi_i$ and repeatedly apply the following transformation $n$ times:
\begin{align*}
    \varphi^{(j+1)}_i &= \frac{\varphi^{(j)}_{2i+1} + \varphi^{(j)}_{2i}}{\sqrt{2}},\\
    \overline{\varphi}^{(j+1)}_i &= \frac{\varphi^{(j)}_{2i+1} - \varphi^{(j)}_{2i}}{\sqrt{2}}.
\end{align*}
As a result one obtains the orthogonal wavelet components $\overline{\varphi}^{(j)}$ at all scales $j=1,\dots, n$. In particular, the coarsest-scale wavelet field has a single component entry
\begin{align*}
  \varphi^{(n)}= \frac{1}{(\sqrt{2})^n}\sum_{i=1}^N \phi_i = \frac{1}{\sqrt{N}}\sum_{i=1}^N \phi_i.
\end{align*}
For the $\phi^4$ model with $\varphi^{(n)}=\sqrt{N}\cdot M$, the coarsest-scale wavelet field is thus the magnetization up to a prefactor. In the numerical experiments, we used $\varphi^{(n)}$ as the collective variable $z$ and the $(\overline{\varphi}^{(j)})_{0 \leq j \leq n}$ as the variables $x^\perp$.
\end{document}